\documentclass[a4paper]{article}
\usepackage[margin = 1in]{geometry}
\usepackage[english,german]{babel}

\usepackage{amsmath, amsfonts, amssymb, amsthm, mathrsfs, physics, braket}

\usepackage{hyperref} 

\hypersetup{colorlinks=true, citecolor=blue, linkcolor=blue, urlcolor=blue} 

\usepackage[title]{appendix}

\newtheorem{definition}{Definition}
\newtheorem{theorem}{Theorem}
\newtheorem{remark}{Remark}
\newtheorem{proposition}{Proposition}
\newtheorem{corollary}{Corollary}
\newtheorem{lemma}{Lemma}

\newcommand{\hilbertspace}{\mathcal{H}}
\newcommand{\operator}[1]{\mathsf{#1}}
\newcommand{\domain}{\mathcal{D}}
\newcommand{\canonicaldomain}{\domain}
\newcommand{\cd}{\canonicaldomain}
\newcommand{\timeinvariantset}{\mathscr{T}}
\newcommand{\tis}{\timeinvariantset}
\newcommand{\reals}{\mathbb{R}}

\title{Canonical pairs in finite-dimensional Hilbert space}
\author{Ralph Adrian E. Farrales* and Eric A. Galapon \\ \\
Theoretical Physics Group, National Institute of Physics \\ University of the Philippines Diliman, Philippines \\ \\
*Corresponding author: refarrales@up.edu.ph}
\date{\today}

\begin{document}

\selectlanguage{english}

\maketitle

\begin{abstract}
    A pair of Hermitian operators is canonical if they satisfy the canonical commutation relation. It has been believed that no such canonical pair exists in finite-dimensional Hilbert space. Here, we obtain canonical pairs by noting that the canonical commutation relation holds in a proper subspace of the Hilbert space. For a given Hilbert space, we study the many possible canonical pairs and look into the uncertainty relation they satisfy. We apply our results by constructing time operators in finite-dimensional quantum mechanics.
    
    \bigskip \noindent
    keywords: canonical pair, canonical commutation relation, finite-dimensional Hilbert space, time operator
\end{abstract}



\section{Introduction}

Traditionally, a canonical pair of observables $A$ and $B$ is represented in quantum mechanics by the pair of self-adjoint operators $\operator{A}$ and $\operator{B}$ on the Hilbert space $\hilbertspace$ satisfying the \textit{canonical commutation relation}
\begin{equation} \label{eq:ccr}
    [\operator{A},\operator{B}]=i\hbar\operator{I}_\hilbertspace \, ,
\end{equation}
where $[\operator{A},\operator{B}] = \operator{AB}-\operator{BA}$ is the commutator, $\hbar$ is the reduced Planck's constant, and $\operator{I}_\hilbertspace$ is the identity operator on $\hilbertspace$. However, it is known that \eqref{eq:ccr} does not hold in finite-dimensional Hilbert space \cite{born1925quantenmechanik,waerden1967sources,weyl1928gruppentheorie,weyl1950groups,putnam1967commutation}, since taking the trace of \eqref{eq:ccr} gives a zero on the left-hand side, and $i\hbar N$ on the right-hand side, where $N$ is the dimension of $\hilbertspace$.

Finding canonical pairs in finite-dimensional Hilbert space has been abandoned ever since. This is in contrast to the infinite-dimensional case, wherein under certain conditions, a canonical pair of self-adjoint operators exists, e.g., the unbounded and completely continuous position and momentum operators. Meanwhile, if one of the pair, say, $\operator{B}$, is bounded or semibounded, then there exists no self-adjoint $\operator{A}$ that can satisfy the canonical commutation relation \cite{pauli1933handbuch,pauli1958handbuch,pauli1980general}. Since a time observable and the Hamiltonian are expected to form a canonical pair, then the semiboundedness of the Hamiltonian caused the longstanding doubt on the existence of a time operator in quantum mechanics. This problem then extends to the finite-dimensional case---a quantum system with a bounded and discrete Hamiltonian offers no possibility for the existence of a time operator such that both operators satisfy \eqref{eq:ccr}.

This problem can be bypassed by noting that the canonical commutation relation holds only in a proper subspace of the Hilbert space, referred to as the canonical domain \cite{galapon2002pauli}. This then allows the possibility of having pairs of bounded Hermitian operators that are canonical. One such example is the existence of (characteristic) time operators that form a canonical pair with the semibounded discrete Hamiltonian (with some growth condition) in a dense subspace of the infinite-dimensional Hilbert space \cite{galapon2002self}. It should be noted that there are infinitely many possible canonical pairs in the same Hilbert space \cite{galapon2006could,farrales2022conjugates}. Different canonical pairs may correspond to different physical observables satisfying different properties. As an example, for a particle in a box, another possible operator forming a canonical pair with the Hamiltonian is the confined passage time, which is the amount of time that has passed as a particle moves from one point to another \cite{galapon2006could}. This is physically different from the characteristic time operator.

We recently showed in \cite{farrales2025characteristic} that the two-dimensional projection of the characteristic time operator acts as the \textit{hand} of a quantum clock, and is a canonical pair with the Hamiltonian in a one-dimensional canonical domain, only for a specific set of times of total measure zero, which we referred to as the time invariant set. This motivates us to revisit the problem of canonical pairs in finite-dimensional Hilbert space. We first give a review on canonical pairs in Section \ref{sec:canonical}. We shall define the canonical domain, the parameter invariant set, and note how different canonical pairs can be categorized as either equivalent or inequivalent. We demonstrate by first looking at the simplest cases: we construct canonical pairs given a two-dimensional Hilbert space in Section \ref{sec:2d}, as well as for the three-dimensional Hilbert space in Section \ref{sec:3d}. We show how equivalent and inequivalent pairs arise for each case, before moving to the general $N$-dimensional case in Section \ref{sec:nd}. We look into the uncertainty relation that the canonical pair satisfies in Section \ref{sec:uncertainty}, find the cases when this relation is saturated, and note how equivalent and inequivalent pairs affect the uncertainty relation. We then apply our results in Section \ref{sec:time} by constructing time operators that form canonical pairs with bounded discrete Hamiltonians. We show that measurement of these time operators give the parametric time when near the (time) parameter invariant set, and thus corresponds to a clock hand observable. We shall also look at how equivalent and inequivalent time operators affect the properties of the quantum clock. We then conclude with Section \ref{sec:conclusions}.

\section{Canonical pairs and commutators} \label{sec:canonical}
Consider the $N$-dimensional Hilbert space $\hilbertspace=\mathbb{C}^N$ for integer $N \geq 2$. Quantum states are elements of $\hilbertspace$ with unit norm. Operators are then $N \times N$ matrices acting on elements of $\hilbertspace$. We first define the commutator \cite{putnam1967commutation}.
\begin{definition}
    An operator $\operator{C}$ is called a commutator if there exists a pair of operators $(\operator{A},\operator{B})$ such that $\operator{C} = [\operator{A},\operator{B}] = \operator{AB} - \operator{BA}$ is nontrivial.
\end{definition}

\noindent
Since we do not include $\operator{C} = 0$ in the definition, then we shall also ignore pairs of operators $(\operator{A},\operator{B})$ where either of them is proportional to the identity.

\begin{proposition}
    Let $\operator{C} = [\operator{A},\operator{B}]$ be a commutator.
    \begin{enumerate}
        \item $\operator{C}$ is traceless.
        \item If the pair $(\operator{A},\operator{B})$ is Hermitian, then $\operator{C}=[\operator{A},\operator{B}]$ is anti-Hermitian.
    \end{enumerate}
\end{proposition}

\noindent
The commutator cannot be proportional to the identity in finite-dimensional Hilbert space \cite{born1925quantenmechanik,waerden1967sources,weyl1928gruppentheorie,weyl1950groups,putnam1967commutation}.

\begin{theorem} \label{thm:nogo}
    Let $\operator{I}_\hilbertspace$ be the identity on $\hilbertspace$. If $\operator{C}$ is a commutator, then it cannot be proportional to $\operator{I}_\hilbertspace$. Conversely, if an operator $\operator{D}$ is proportional to $\operator{I}_\hilbertspace$, then it is never a commutator.
\end{theorem}
\begin{proof}
    Let $\operator{C}$ be a commutator and let $\operator{D}$ be proportional to $\operator{I}_\hilbertspace$, i.e., $\operator{D} = d\operator{I}_\hilbertspace$ for nonzero constant $d$. Suppose for some $d$, $\operator{C} = \operator{D}$. However, since $\hilbertspace=\mathbb{C}^N$ and $N\neq0$, then $\tr(\operator{C}) = 0 \neq \tr(\operator{D}) = dN$, a contradiction.
\end{proof}

\noindent 
Thus, given a commutator $\operator{C}$ and nonzero constant $c$, the relation $\operator{C}\ket{\varphi} = c\ket{\varphi}$ cannot hold for all $\ket{\varphi}\in\hilbertspace$. The canonical commutation relation then only holds in a proper subspace of the finite-dimensional Hilbert space, referred to as the \textit{canonical domain} \cite{galapon2002pauli}.

\begin{definition}
    An ordered pair of Hermitian (self-adjoint) operators $(\operator{A},\operator{B})$ on $\hilbertspace$ is referred to as a canonical pair if it satisfies the canonical commutation relation
    \begin{equation} \label{eq:ccrd}
        [\operator{A},\operator{B}] \ket{\varphi} = i\hbar \ket{\varphi} ,
    \end{equation}
    for all $\ket{\varphi}$ in some (nontrivial and nonempty) proper subspace $\cd$ of $\hilbertspace$ referred to as the canonical domain. The triple $\mathcal{C}(\operator{A},\operator{B},\cd)$ is referred to as a solution to the canonical commutation relation if $(\operator{A},\operator{B})$ is a canonical pair in $\cd$. 
\end{definition}

\noindent 
In \cite{galapon2006could}, an infinite-dimensional Hilbert space was considered, and thus, the canonical domain can be dense or closed. A solution to the (infinite-dimensional) canonical commutation relation $\mathcal{C}(\operator{A},\operator{B},\cd)$ was referred to as a dense-category solution if $\cd$ is dense; otherwise, it is referred to as a closed-category solution. In the finite-dimensional case, the canonical domain is always closed, and thus, the solutions are always of closed-category.

We can define the more general case.

\begin{definition}
    An ordered pair of operators $(\operator{A},\operator{B})$ on $\hilbertspace$ is referred to as an essentially canonical pair (with respect to $c$) if it satisfies the essentially canonical commutation relation
    \begin{equation} \label{eq:eccrd}
        [\operator{A},\operator{B}] \ket{\varphi} = c \ket{\varphi} ,
    \end{equation}
    where $c$ is complex and nonzero, for all $\ket{\varphi}$ in some (nontrivial and nonempty) proper subspace $\cd$ of $\hilbertspace$ referred to as the (essentially) canonical domain. $\mathcal{C}^c(\operator{A},\operator{B},\cd)$ is referred to as a solution to the essentially canonical commutation relation if $(\operator{A},\operator{B})$ is an essentially canonical pair in $\cd$ with respect to $c$.
\end{definition}


\noindent
Every canonical pair is thus an essentially canonical pair. Meanwhile, every Hermitian essentially canonical pair can be mapped to a canonical pair. A Hermitian pair gives an anti-Hermitian commutator, and thus the commutator eigenvalues are purely imaginary. Suppose we have the Hermitian essentially canonical pair $(\operator{A},\operator{B})$ satisfying $[\operator{A},\operator{B}] \ket{\varphi} = ic \ket{\varphi}$ for $\ket{\varphi}\in\cd$ and real $c$. There are infinitely many ways to generate canonical pairs from this: we have the canonical pair $(\hbar^{\lambda} \operator{A} / c^{\rho} , \hbar^{1-\lambda} \operator{B} / c^{1-\rho} )$ in the same canonical domain $\cd$, where $\lambda$ and $\rho$ are constants, giving us the canonical solution $\mathcal{C}(\hbar^{\lambda} \operator{A} / c^{\rho} , \hbar^{1-\lambda} \operator{B} / c^{1-\rho},\cd)$. For example, if we have the essentially canonical relation $[\operator{A},\operator{B}]\ket{\varphi}=-i\hbar\ket{\varphi}$ for $\ket{\varphi}\in\cd$, one canonical solution that we can generate from this is $\mathcal{C}(-\operator{A},\operator{B},\cd)$.

\begin{remark}
    Consider the commutator $\operator{C}=[\operator{A},\operator{B}]$ satisfying the essentially canonical relation $\operator{C}\ket{\varphi} = c\ket{\varphi}$ for $\ket{\varphi}$ in its canonical domain $\cd$. Then, $\ket{\varphi}$ is an eigenket of $\operator{C}$ corresponding to the eigenvalue $c$, where the canonical domain is the eigenspace corresponding to the eigenvalue $c$.
\end{remark}

\begin{theorem} \label{thm:degen}
    A commutator on $\hilbertspace$ has degeneracy of at most $N-1$.
\end{theorem}
\begin{proof}
    For $\hilbertspace=\mathbb{C}^N$, the commutator is an $N\times N$ matrix, and its characteristic polynomial has $N$ roots according to the fundamental theorem of algebra, and these roots correspond to the eigenvalues of the commutator. The multiplicity of each eigenvalue is its degeneracy. If it has an eigenvalue $c$ with degeneracy equal to $N$, then $c$ has to be zero for the commutator to be traceless, which corresponds to the trivial matrix. 
\end{proof}

\begin{proposition}
    A commutator on $\hilbertspace$ has at least two nonzero eigenvalues.
\end{proposition}
\begin{proof}
    The zero eigenvalue (if it exists) of a commutator cannot have degeneracy of $N$ by definition. In fact, it also cannot have degeneracy of $N-1$. Suppose the commutator has two distinct eigenvalues, the zero eigenvalue and some nonzero $c$. If the zero eigenvalue has degeneracy $N-1$, then $c$ has multiplicity equal to one, but since $c$ is nonzero, then that contradicts with the fact that the commutator is traceless. Thus, the zero eigenvalue has degeneracy of at most $N-2$, and there exists at least two nonzero eigenvalues.
\end{proof}

\noindent
That is, the commutator with the smallest number of nonzero eigenvalues is of the form 
\begin{equation} \label{eq:commex}
    \operator{C} = \mqty(\dmat{0, \ddots, 0, c, -c }) \, ,
\end{equation}
for nonzero constant $c$. This means that the commutator of two operators always satisfies at least two different essentially canonical commutation relations.

\begin{proposition} \label{prop:eccr}
    Every noncommuting pair of operators on $\hilbertspace$ is an essentially canonical pair, satisfying at least two different essentially canonical commutation relations.
\end{proposition}
\begin{proof}
    For noncommuting $\operator{A}$ and $\operator{B}$, i.e., $\operator{C} = [\operator{A},\operator{B}] \neq 0$, then $\operator{C}$ is nontrivial, traceless, and contains at least two distinct nonzero eigenvalues $c_1$ and $c_2$. Each eigenvalue corresponds to eigenfunctions in their respective eigenspaces $\cd_1$ and $\cd_2$, wherein $\cd_1 \neq \cd_2$. We then have two essentially canonical commutation relations: $\operator{C}\ket{\varphi} = [\operator{A},\operator{B}]\ket{\varphi} = c_1\ket{\varphi}$ for $\ket{\varphi}\in\cd_1$ and $\operator{C}\ket{\varphi}=c_2\ket{\varphi}$ for $\ket{\varphi} \in \cd_2$.
\end{proof}

\noindent
This means that a single pair of noncommuting operators may behave in at least two different ways, depending on whether a particular state is in one canonical domain or the other. Suppose that there exists a pair $(\operator{A},\operator{B})$ whose commutator is of the form \eqref{eq:commex} (with respect to the eigenkets of $\operator{C}$). Then, the eigenspaces of the nonzero eigenvalues have dimension equal to one. With the commutator satisfying $\operator{C}\ket{\varphi_\pm} = \pm c \ket{\varphi_\pm}$ for $\ket{\varphi_\pm} \in \cd_\pm$, then the canonical domains $\cd_\pm$ are also one-dimensional. Firstly, we see that the same pair (corresponding to the same commutator) behave differently depending on whether the state is in $\cd_+$ or in $\cd_-$. Secondly, we see that the number of zero eigenvalues of the commutator affects the maximum dimension of the eigenspaces, and thus, of the canonical domains.

Next, we have the following theorem by \cite{galapon2002pauli}.

\begin{theorem} \label{thm:noeigenket}
    Consider a Hermitian essentially canonical pair $(\operator{A},\operator{B})$ with canonical domain $\cd$. The eigenkets of $\operator{A}$ and $\operator{B}$ do not belong to $\cd$.
\end{theorem}
\begin{proof}
    Consider $[\operator{A},\operator{B}] = i c\operator{I}_{\cd}$ where $c$ is nonzero and real, and $\operator{I}_{\cd}$ is the identity in $\cd$, i.e., $\operator{I}_\cd \ket{\varphi} = \ket{\varphi}$ for all $\ket{\varphi} \in \cd$. Let $\operator{\operator{A}}\ket{\phi}=a\ket{\phi}$. Suppose that its eigenkets $\ket{\phi}$ are in $\cd$. Then the eigenspace of $\operator{A}$ is invariant under the commutator $[\operator{A},\operator{B}]$. Since $\operator{A}$ is Hermitian, we then have $\braket{\phi|[\operator{A},\operator{B}]\phi}=0$, i.e., $[\operator{A},\operator{B}]$ maps each eigenket of $\operator{A}$ to its orthogonal. But this contradicts with the fact that $(ic)^{-1}[\operator{A},\operator{B}]$ is the identity in $\cd$. The same can be said about the eigenkets of $\operator{B}$.
\end{proof}

We showed in \cite{farrales2025characteristic} that the time-evolution operator $\exp(-i\operator{H}t/\hbar)$ does not necessarily leave $\cd$ invariant. That is, $\ket{\varphi(t)} = \exp(-i\operator{H}t/\hbar)\ket{\varphi}$ is not necessarily in $\cd$ for all $t$, and are only in $\cd$ for a set of times referred to as the time invariant set. We then define the \textit{parameter invariant set} for a general unitary matrix $\operator{U}(t)$.

\begin{definition}
    Consider a solution $\mathcal{C}^c(\operator{A},\operator{B},\cd)$ to the essentially canonical commutation relation. Consider a one-parameter unitary matrix $\operator{U}(t)$. For all $\ket{\varphi}$ in $\cd$, the parameter invariant set $\tis \subseteq \reals$ is the set of all $t$ such that $\operator{U}(t)\ket{\varphi}$ is also in $\cd$. Equivalently, for $\operator{A}(t) = \operator{U}^\dagger(t)\operator{A}\operator{U}(t)$ and $\operator{B}(t) = \operator{U}^\dagger(t)\operator{B}\operator{U}(t)$, then the the parameter invariant set is the set of all $t$ such that $\mathcal{C}^c(\operator{A}(t),\operator{B}(t),\cd)$ is a solution to the essentially canonical commutation relation.
\end{definition}

\noindent 
For $\operator{U}(t = 0)\ket{\varphi} = \ket{\varphi}$, then $t = 0$ is always an element of $\tis$, and thus, the parameter invariant set is never empty. With the unitary time-evolution operator $\operator{U}(t) = \exp(-i\operator{H}t/\hbar)$ where $\operator{H}$ is the Hamiltonian and $t$ is the parametric time, then the time operator and the Hamiltonian form a canonical pair for the set of times $t$ given by the parameter invariant set (referred to as the \textit{time} invariant set in \cite{farrales2025characteristic}). Throughout the paper, the $\operator{H}$ in $\exp(-i\operator{H}t/\hbar)$ need not necessarily be the Hamiltonian of the system, though we can still think of $\exp(-i\operator{H}t/\hbar)$ as an evolution operator with $t$ as the evolution parameter.

\begin{definition}
    $\mathcal{C}^c(\operator{A},\operator{B},\cd;\operator{U}(t),\tis)$ is denoted as the class of solutions $\mathcal{C}^c(\operator{U}^\dagger(t)\operator{A}\operator{U}(t),\operator{U}^\dagger(t)\operator{B}\operator{U}(t),\cd)$ to the essentially canonical commutation relation for all $t \in \tis$.
\end{definition}

\begin{theorem} \label{thm:tis}
    Consider $\mathcal{C}^c(\operator{A},\operator{B},\cd;\exp(-i\operator{H}t/\hbar),\tis)$ as a class of solutions to the essentially canonical commutation relation, where $\operator{A}$, $\operator{B}$, and $\operator{H}$ are Hermitian and $t$ is real. 
    If the canonical domain $\cd$ is one-dimensional and if $\ket{\varphi} \in \cd$ is an eigenket of $\operator{H}$, then $\tis=\reals$. Otherwise, $\tis$ is a set of measure zero.
\end{theorem}
\begin{proof}
    Suppose $\cd$ is one-dimensional, $\ket{\varphi} \in \cd$, and $\operator{H}\ket{\varphi} = E\ket{\varphi}$. Then, $\ket{\varphi(t)}=\exp(-i\operator{H}t/\hbar)\ket{\varphi} = \exp(-iEt/\hbar)\ket{\varphi}$ is in $\cd$ for all $t$. Thus, $\tis = \mathbb{R}$.

    Suppose $\cd$ is one-dimensional, $\ket{\varphi} \in \cd$, and $\ket{\varphi}$ is not an eigenket of $\operator{H}$. Let $\operator{H}\ket{s} = E_s \ket{s}$. We can write $\ket{\varphi} = \sum_{s=1}^N \varphi_s \ket{s}$, where at least two constant $\varphi_s$ are nonzero. Let $\delta(\operator{H})$ be the set of all absolute differences of the eigenvalues of $\operator{H}$, i.e., the set of all $\abs{E_k - E_l}$ for all $k$ and $l$. Then, $\ket{\varphi(t)}=\exp(-i\operator{H}t/\hbar)\ket{\varphi} = \sum_{s=1}^N \varphi_s \exp(-iE_st/\hbar)\ket{s}$ is in $\cd$ only when $t = 2n\pi\hbar/\gcd(\delta(\operator{H}))$ for $n \in \mathbb{Z}$, where $\gcd$ is the greatest common divisor. Therefore, $\tis$ is a set of measure zero.

    Suppose elements of an $M$-dimensional canonical domain can be written as $\ket{\varphi}=\sum_{k=1}^M \alpha_k \ket{\psi_k}$ where $M > 1$, $\alpha_k$ is complex and nonzero, and $\ket{\psi_k}$ is not necessarily an eigenket of $\operator{H}$. Since $\operator{H}$ is Hermitian, then its eigenkets are orthogonal. We can always find two elements in the canonical domain that are not orthogonal: if for example $\ket{\psi_1}$ is an eigenket of $\operator{H}$, then there exists an element of the canonical domain, say, $2^{-1/2}(\ket{\psi_1}+\ket{\psi_2})$ that cannot be its eigenket. Thus, there is at least one $\ket{\varphi}$ that is not an eigenket of $\operator{H}$. Let $\operator{H}\ket{s} = E_s \ket{s}$. We can write $\ket{\varphi} = \sum_{s=1}^N \varphi_s \ket{s}$, and if $\ket{\varphi}$ is not an eigenket of $\operator{H}$, then at least two constant $\varphi_s$ are nonzero. We have at least one element of $\cd$, $\ket{\varphi(t)}=\exp(-i\operator{H}t/\hbar)\ket{\varphi} = \sum_{s=1}^N \varphi_s \exp(-iE_st/\hbar)\ket{s}$, that can only be in $\cd$ in a set of measure zero. Thus, for all $\ket{\varphi} \in \cd$, we see that $\ket{\varphi(t)}$ is guaranteed to be in $\cd$ only in a set of measure zero.
\end{proof}

\begin{remark}
    For the class of solutions $\mathcal{C}^c(\operator{A},\operator{B},\cd;\exp(-i\operator{H}t/\hbar),\tis)$, if $[\operator{A},\operator{H}]=0$ or $[\operator{B},\operator{H}]=0$, then the parameter invariant set $\tis$ has total measure zero. This is because $\operator{H}$ has a common set of eigenkets with either $\operator{A}$ or $\operator{B}$, but their eigenkets are not in $\cd$ (Theorem \ref{thm:noeigenket}).
\end{remark}

\begin{remark}
    The greatest common divisor, $\gcd$, of a set of real numbers, $\delta = \qty{\delta_1, \delta_2, \dotsc}$, is the largest real number $\delta_0$ such that we can write each element of $\delta$ as $\delta_k = m_k \delta_0$ for integer $m_k$, i.e., $\gcd(\delta)=\delta_0$. Suppose that the parameter invariant set is $\tis = \qty{t=n/\gcd(\delta),n\in \mathbb{Z}}$. If there exists two nonzero elements of $\delta$ that are incommensurate, i.e., $\delta_j/\delta_l$ is irrational, then there exists no $\gcd(\delta)$, and we take the parameter invariant set to be $\tis = \qty{t=0}$. That is, let $(\operator{A},\operator{B})$ be a canonical pair with canonical domain $\cd$, let $\tis = \qty{t=0}$ be the parameter invariant set of the unitary $\operator{U}(t)$, and let $\operator{A}(t)= \operator{U}^\dagger(t)\operator{A}\operator{U}(t)$ and $\operator{B}(t)= \operator{U}^\dagger(t)\operator{B}\operator{U}(t)$. Then, $(\operator{A}(t),\operator{B}(t))$ is a canonical pair in $\cd$ only when $t = 0$.
\end{remark}

The canonical commutation relation between the time operator and the Hamiltonian is an example of having the unitary time-evolution operator possessing a (time) parameter invariant set that is of total measure zero \cite{farrales2025characteristic}. For an example where the parameter invariant set is the entire real line, consider the Pauli matrices, $\sigma_x$, $\sigma_y$, and $\sigma_z$, which are Hermitian operators on the two-dimensional Hilbert space. We have the Hermitian essentially canonical pair $(\sigma_x,\sigma_y)$ satisfying the essentially canonical commutation relation $[\sigma_x,\sigma_y]\ket{\sigma_z,+}=2i\ket{\sigma_z,+}$, where $\ket{\sigma_z,+}$ is an element of the one-dimensional canonical domain $\cd$, and where $\sigma_z\ket{\sigma_z,\pm}=\pm\ket{\sigma_z,\pm}$. Consider the unitary operator $\operator{U}(t) = \exp(-i\sigma_z t)$. Then, we have the class of solutions $\mathcal{C}^{2i}(\sigma_x,\sigma_y,\cd; \exp(-i\sigma_z t), \tis)$ with $\tis = \mathbb{R}$ (Theorem \ref{thm:tis}).

Two different unitary matrices $\operator{U}_1(t)$ and $\operator{U}_2(t)$ give the two classes of solutions $\mathcal{C}^c_1(\operator{A},\operator{B},\cd;\operator{U}_1(t),\tis_1)$ and $\mathcal{C}^c_2(\operator{A},\operator{B},\cd;\operator{U}_2(t),\tis_2)$. These two classes are encompassed under the two-parameter unitary operator $\operator{U}(t_1,t_2)=\operator{U}_1(t_1)\operator{U}_2(t_2)$ for $t_1\in\tis_1$ and $t_2\in\tis_2$, and thus, one can unitarily transform the canonical pair $(\operator{A},\operator{B})$ to either class.

\begin{definition}
    Two solutions $\mathcal{C}^c_1(\operator{A}_1,\operator{B}_1,\cd_1)$ and $\mathcal{C}^c_2(\operator{A}_2,\operator{B}_2,\cd_2)$ to the same essentially canonical commutation relation are equivalent if there exists a unitary operator $\operator{U}$ such that $\operator{A}_2 = \operator{U}^\dagger\operator{A}_1\operator{U}$, $\operator{B}_2 = \operator{U}^\dagger\operator{B}_1\operator{U}$, and $\ket{\varphi_2} = \operator{U}^\dagger \ket{\varphi_1}$ for every $\ket{\varphi_1}\in\cd_1$ and $\ket{\varphi_2}\in\cd_2$. Otherwise, they are inequivalent solutions.
\end{definition}

\begin{remark}
    The class of solutions $\mathcal{C}(\operator{A},\operator{B},\cd;\operator{U}(t),\tis)$ are equivalent solutions to the canonical commutation relation.
\end{remark}

\noindent 
In the same Hilbert space, there exists infinitely many canonical pairs. Given the canonical pair $(\operator{A},\operator{B})$ with canonical domain $\cd$, we can obtain other canonical pairs by finding its unitary equivalents. Another way of obtaining other solutions is by considering an operator $\operator{F}$ satisfying $[\operator{F},\operator{B}]=0$ giving us the canonical pair $(\operator{A}+\operator{F},\operator{B})$ in the same $\cd$. In general, it is possible to find two operators $\operator{A}_1$ and $\operator{A}_2$ forming a canonical pair with the same $\operator{B}$ such that $\mathcal{C}_1(\operator{A}_1,\operator{B},\cd_1)$ and $\mathcal{C}_2(\operator{A}_2,\operator{B},\cd_2)$ are inequivalent solutions. One of us \cite{galapon2006could} gave such an example (in infinite-dimensional Hilbert space): given the Hamiltonian of a particle in a box, there exists a dense-category solution given by the characteristic time operator \cite{galapon2002self} and a closed-category solution given by the confined passage time operator. Both are physically different, the former being related to a quantum clock \cite{farrales2025characteristic}, the latter representing how much time has passed for a particle moving from one point to another.

\begin{remark}
    Inequivalent solutions to the canonical commutation relation correspond to canonical pairs that are physically different, satisfying different properties.
\end{remark}

\noindent
Finally, different essentially canonical pairs may also correspond to pairs that are physically different. For example, given the Hamiltonian $\operator{H}$, suppose we have the pair $(\operator{T}_\pm,\operator{H})$ which satisfies two different essentially canonical commutation relations $[\operator{T}_\pm,\operator{H}]\ket{\varphi_\pm}=\pm i\hbar \ket{\varphi_\pm}$ for $\ket{\varphi_\pm}\in\cd_\pm$. Then, the canonical pair $(\operator{T}_+,\operator{H})$ corresponds to an operator $\operator{T}_+$ that increases in step with time for $\ket{\varphi_+}\in\cd_+$, while the other pair corresponds to a $\operator{T}_-$ that would instead decrease in step with time for $\ket{\varphi_-}\in\cd_-$.

Therefore, for a given Hermitian operator $\operator{B}$, the objective for the next few sections is to find a Hermitian $\operator{A}$ such that $(\operator{A},\operator{B})$ form a canonical pair in some canonical domain. From there, we naturally obtain (Hermitian) essentially canonical pairs, as well as equivalent and inequivalent solutions. We first consider a nondegenerate operator $\operator{B}$ with eigenkets $\ket{s}$ with the corresponding real eigenvalues $B_s$, $s = 1, 2, \dotsc, N$. We let $-\infty<B_s<B_{s+1}<\infty$. We also assume that its eigenkets form a complete orthonormal basis, such that we can write the $N$-dimensional Hilbert space as $\hilbertspace = \qty{\ket{\varphi}=\sum_{s=1}^N \varphi_s \ket{s}, \varphi_s \in \mathbb{C} }$. For $\ket{\varphi}$ in $\hilbertspace$, the canonical commutation relation \eqref{eq:ccrd} gives
\begin{equation} \label{eq:Aklchareq}
    \sum_{k=1}^N \varphi_k B_{kl} A_{lk} = i\hbar \varphi_l \, ,
\end{equation}
where $B_{kl} = B_k - B_l$ and $A_{lk} = \braket{l|\operator{A}|k}$, with $B_{kl}=-B_{lk}$ and $A_{lk} = A^*_{kl}$ (since $\operator{A}$ is also assumed to be Hermitian). For a given $\operator{B}$, one can solve for the matrix elements of $\operator{A}$ to obtain the canonical pair $(\operator{A},\operator{B})$. One can then solve the eigenspace of the commutator $[\operator{A},\operator{B}]$ to determine the corresponding canonical domain $\cd$. Essentially canonical pairs can be obtained by looking at other eigenvalues of the commutator aside from $i\hbar$. Finally, one looks at how some $\operator{U}(t)$ affects $\cd$, and find the values of $t$ wherein $\operator{U}(t)\ket{\varphi}\in\cd$ for $\ket{\varphi}\in\cd$ to obtain the parameter invariant set $\tis$.

We shall also consider the case when $\operator{B}$ contains degeneracy, where we have $\ket{s,r}$ as the eigenket corresponding to the eigenvalue $B_s$ for $s = 1, 2, \dotsc, L$ and $r = 1, 2, \dotsc, M_s$ for positive integer $L$ and $M_s$, i.e., the eigenvalue $B_s$ has degeneracy $M_s$. We once again assume that the eigenkets form a complete orthonormal basis such that for $N = \sum_{s=1}^L M_s$, the $N$-dimensional Hilbert space can explicitly be written as $\hilbertspace=\qty{\ket{\varphi}=\sum_{s=1}^L \sum_{r=1}^{M_s} \varphi_{s,r} \ket{s,r}, \varphi_{s,r}\in\mathbb{C}}$. For $\ket{\varphi}$ in $\hilbertspace$, \eqref{eq:ccrd} gives
\begin{equation} \label{eq:Askdegenchareq}
    \sum_{s'=1}^L \sum_{r'=1}^{M_{s'}} \varphi_{s',r'} A_{s,r,s',r'} B_{s,s'} = i\hbar \varphi_{s,r} \, ,
\end{equation}
where $A_{s,r,s',r'} = \braket{s,r|\operator{A}|s',r'}$ and $B_{s,s'} = B_s - B_{s'}$. One sees that there exists no solution when $B$ is purely degenerate (when L = 1).

\begin{proposition}
    Consider a Hermitian $\operator{B}$ such that $\operator{B} \ket{s,r} = B_s \ket{s,r}$ for $s = 1, \dotsc, L$ and $r = 1, \dotsc, M_s$. If $L = 1$, then no Hermitian $\operator{A}$ can form a canonical pair with $\operator{B}$. 
\end{proposition}
\begin{proof}
    If $L = 1$, then the $s$ and $s'$ in \eqref{eq:Askdegenchareq} can only be equal to 1. Thus, $B_{s,s'} = B_s - B_{s'}$ is always zero, and so is the left-hand side of \eqref{eq:Askdegenchareq}. We then have $\varphi_{s,r} = 0$ for all $s$ and $r$, i.e., the canonical relation is only satisfied when $\ket{\varphi} = 0$, corresponding to a trivial subspace of $\hilbertspace$.
\end{proof}

\noindent
We thus have the condition that $M_s < N$ for all $s$, i.e., the degeneracies are always smaller than the dimension of the Hilbert space. In two-dimensional Hilbert space, there then does not exist a canonical pair $(\operator{A},\operator{B})$ if $\operator{B}$ is purely degenerate. Therefore, the simplest case allowing canonical pairs $(\operator{A},\operator{B})$ where $\operator{B}$ has degeneracy is in a three-dimensional Hilbert space with one eigenvalue of $\operator{B}$ being doubly-degenerate.

\section{Canonical pairs in two-dimensional Hilbert space} \label{sec:2d}
We demonstrate here for the 2D case. For $N=2$ and for a nondegenerate $\operator{B}$, \eqref{eq:Aklchareq} gives
\begin{equation}
    (i\hbar)^2 + B_{12}^2 \abs{A_{12}}^2 = 0 \, ,
\end{equation}
where $B_{12} = B_1 - B_2$ and $A_{12}=\braket{1|\operator{A}|2}$, giving the general solution $A_{12} = \hbar \exp(i\alpha)/\abs{B_{12}} = A_{21}^*$ for real $\alpha$, with $A_{11}=a_1$ and $A_{22}=a_2$ being any real number. With respect to the eigenkets of $\operator{B}$, $\ket{1} = \smqty(1 \\ 0)$ and $\ket{2} = \smqty(0 \\ 1)$, we then have the canonical pair in matrix form
\begin{equation} \label{eq:2dcp}
    \operator{A}_\alpha = \mqty( a_1 && -\frac{\hbar}{B_{12}}e^{i\alpha} \\ -\frac{\hbar}{B_{12}}e^{-i\alpha} && a_2 ) \, , \qquad \operator{B} = \mqty(B_1 && 0 \\ 0 && B_2) \, .
\end{equation}

Notice that $\operator{A}_\alpha = \exp(i\alpha \operator{B}/B_{12})\operator{A}_{\alpha=0}\exp(-i\alpha \operator{B}/B_{12})$, thus for fixed $a_1$ and $a_2$, the solutions $\operator{A}_\alpha$ for different $\alpha$ are just unitary transformations of each other. Notice also that we can rewrite $\operator{A}_\alpha$ into
$\smqty( 0 && -\frac{\hbar}{B_{12}}e^{i\alpha} \\ -\frac{\hbar}{B_{12}}e^{-i\alpha} && 0 ) + \smqty(a_1 && 0 \\ 0 && a_2)$, and see that the second term commutes with $\operator{B}$. Note that varying $a_1$ and $a_2$ could change the spectral properties of $\operator{A}_\alpha$ (such as its eigenvalues or its trace). Since unitary transformations preserve spectral properties, then the canonical pairs for various $a_1$ and $a_2$ are not necessarily unitary equivalents of each other. This then provides a possible way (which for the 2D case, is also the only way) of constructing inequivalent solutions to the 2-dimensional canonical commutation relation.

We then look at the commutator eigenvalue problem
\begin{equation} \label{eq:2dcomm}
    [\operator{A}_\alpha,\operator{B}]\ket{\varphi} = \hbar \mqty(0 && e^{i\alpha} \\ -e^{-i\alpha} && 0)\ket{\varphi} .
\end{equation}
We solve the eigenvalue problem \eqref{eq:2dcomm}, giving eigenvalues $\pm i\hbar$, corresponding to a canonical pair and an essentially canonical pair, respectively. We obtain a canonical pair by looking at the $+i\hbar$ case, which has the corresponding eigenvector $c \smqty(1 \\ i\exp(-i\alpha))$ for $c \in \mathbb{C}$. Normalization gives  $c = 2^{-1/2} \exp(i\beta)$ for real $\beta$ (though we will ignore the phase shift $\exp(i\beta)$ going forward). We then have the 1-dimensional canonical domain 
\begin{equation} \label{eq:2dcd}
    \cd_\alpha = \qty{ \frac{1}{\sqrt{2}} \mqty(1 \\ ie^{-i\alpha}) } \, ,
\end{equation}
giving us the solution $\mathcal{C}_\alpha(\operator{A}_\alpha,\operator{B},\cd_\alpha)$. The solutions for various $\alpha$ constitute equivalent solutions to the canonical commutation relation. Note that $a_1$ and $a_2$ do not affect the canonical domain, and thus for the 2D case, equivalent or inequivalent solutions arising from adding a term that commutes with $\operator{B}$ do not affect $\cd_\alpha$.

The unitary $\exp(-it\operator{B}/B_{12})$ does not leave the canonical domain invariant. For fixed $\alpha$, this evolution operator pushes elements of $\cd$ outside: since eigenkets of $\operator{B}$ are not in $\cd$ (Theorem \ref{thm:noeigenket}), then $\exp(-it\operator{B}/B_{12})$ moves $\ket{\varphi}\in\cd$ back to $\cd$ only when $t$ is in $\tis = \qty{2\pi n, n \in \mathbb{Z}}$, which is a set of measure zero (Theorem \ref{thm:tis}). We then have the class of solutions $\mathcal{C}_\alpha(\operator{A}_\alpha,\operator{B},\cd_\alpha;\exp(-it\operator{B}/B_{12}),\timeinvariantset)$ giving us the canonical pair and its unitary equivalents for real $\alpha$ and for $t \in \tis$. A different unitary operator generates other equivalent solutions. For the unitary $\operator{U}(t)=\exp(-it\operator{H}/\hbar)$, we have a parameter invariant set $\tis$ being equal to the entire real line if $\operator{H}$ does not commute with $\operator{A}$ or $\operator{B}$ and $\ket{\varphi}\in\cd$ is an eigenket of $\operator{H}$; otherwise $\tis$ has total measure zero.

From the commutator, we also naturally obtain an essentially canonical pair satisfying $[\operator{A}_\alpha,\operator{B}]\ket{\varphi}=-i\hbar\ket{\varphi}$ for $\ket{\varphi}\in\cd_\alpha^-=\qty{2^{-1/2}\smqty(1 \\ -ie^{-i\alpha})}$, giving us the essentially canonical solution $\mathcal{C}_\alpha^{-i\hbar}(\operator{A}_\alpha,\operator{B},\cd_\alpha^-)$. Thus, the same pair $(\operator{A}_\alpha,\operator{B})$ defines two different essentially canonical commutation relations, depending on whether it is in $\cd_\alpha$ or $\cd_\alpha^-$. The essentially canonical solution $\mathcal{C}_\alpha^{-i\hbar}(\operator{A}_\alpha,\operator{B},\cd_\alpha^-)$ can be mapped into a \textit{canonical} solution $\mathcal{C}_\alpha(-\operator{A}_\alpha,\operator{B},\cd_\alpha^-)$: with the unitary operator $\operator{U}(t) = \exp(-it\operator{B}/B_{12})$, we see that this is equivalent to our above canonical solution $\mathcal{C}_\alpha(\operator{A}_\alpha,\operator{B},\cd_\alpha)$  since $-\operator{A}_\alpha = \operator{U}^\dagger(t=\pi)\operator{A}_\alpha\operator{U}(t=\pi)$ and $\operator{U}^\dagger(t=\pi)$ maps $\cd_\alpha$ to $\cd_\alpha^-$ (up to some phase factor). Therefore, $\mathcal{C}_\alpha(-\operator{A}_\alpha,\operator{B},\cd_\alpha^-)$ and $\mathcal{C}_\alpha(\operator{A}_\alpha,\operator{B},\cd_\alpha)$ are equivalent solutions.

From Theorem \ref{thm:degen}, we have the following Corollary.

\begin{corollary}
    The commutator in two-dimensional Hilbert space is nondegenerate.
\end{corollary}

\begin{remark}
    In two-dimensional Hilbert space, commutators have eigenvalues $\pm c$ (since it is traceless) for real nonzero $c$. Thus, we always get two essentially canonical pairs (Proposition \ref{prop:eccr}), both with one-dimensional canonical domains. The pair $(\operator{A},\operator{B})$ being an essentially canonical pair with respect to $-c$ with canonical domain $\cd_-$ can be mapped to $(-\operator{A},\operator{B})$ being an essentially canonical pair with respect to $+c$ in the same $\cd_-$. We see above that for $c = i\hbar$ the two solutions are equivalent. This is also true for any constant $c$.
\end{remark}

For $\ket{s}$ being eigenkets of $\operator{B}$, recall that $\braket{s|[\operator{A},\operator{B}]s}=0$ since $\operator{B}$ is Hermitian. This means that the diagonal of the commutator (with respect to the eigenkets of $\operator{B}$) consists purely of zeros. If we can find a unitary operator transforming any (nontrivial) traceless matrix into a form whose diagonal consists purely of zeros, then we can always construct it as a commutator.

\begin{proposition} \label{prop:2dtracecomm}
    Every nontrivial $2 \times 2$ traceless matrix is a commutator.
\end{proposition}
\begin{proof}
    Suppose $\operator{C}$ is traceless and $\operator{C}\neq 0$. Then it has eigenvalues $c$ and $-c$, for constant nonzero $c$. To the eigenvalue $\pm c$ is a corresponding eigenket $\ket{\varphi_1}$ and $\ket{\varphi_2}$, respectively. Explicitly, $\operator{C} = c \op{\varphi_1} - c \op{\varphi_2}$. There exists a unitary operator $\operator{U}$, e.g., $\operator{U} = 2^{-1/2} \smqty(1 && 1 \\ 1 && -1)$, such that in another basis $\ket{1} = \operator{U}\ket{\varphi_1}$ and $\ket{2} = \operator{U}\ket{\varphi_2}$, the diagonal of $\operator{C}$ purely consists of zeros, e.g., $\operator{U}^\dagger\operator{C}\operator{U} = \smqty(0 && c \\ c && 0)$.

    Next, we construct the operator $\operator{B} = B_1 \ket{1}\bra{1} + B_2 \ket{2}\bra{2}$ where $B_1$ and $B_2$ are constants. We then have the operator
    \begin{equation}
        \operator{A} = a_1 \ket{1}\bra{1} +\frac{\braket{1|\operator{C}|2}}{B_1 - B_2}\ket{1}\bra{2} + \frac{\braket{2|\operator{C}|1}}{B_1 - B_2} \ket{2}\bra{1} + a_2\ket{2}\bra{2} \, ,
    \end{equation}
    where $a_1$ and $a_2$ are constants, such that $\operator{U}^\dagger\operator{C}\operator{U} = [\operator{A},\operator{B}]$. Therefore, $\operator{C}=[\operator{U}\operator{A}\operator{U}^\dagger,\operator{U}\operator{B}\operator{U}^\dagger]$ is a commutator.
\end{proof}
Later, we will generalize this to $N$ dimensions and obtain the same conclusion as \cite{shoda1936einige} regarding traceless matrices. For now, we conclude this section by noting that any element of the two-dimensional Hilbert space corresponds to an essentially canonical pair whose one-dimensional canonical domain contains that element.

\begin{proposition} \label{prop:2devery}
    To every element $\ket{\varphi}$ of the two-dimensional Hilbert space $\hilbertspace$, there exists a commutator $\operator{C}$ such that $\operator{C}\ket{\varphi}=c\ket{\varphi}$, where $c$ is a nonzero constant.
\end{proposition}
\begin{proof}
    For a given $\ket{\varphi}$, consider $\operator{C} = c \op{\varphi} - c \op{\phi}$ wherein $\braket{\varphi|\phi}=0$ and $c \neq 0$. Thus, $\operator{C}\ket{\varphi} = c\ket{\varphi}$. From Proposition \ref{prop:2dtracecomm}, since $\operator{C}$ is traceless, then it is a commutator.
\end{proof}

\section{Canonical pairs in three-dimensional Hilbert space} \label{sec:3d}
For $N=3$ and nondegenerate $\operator{B}$, \eqref{eq:Aklchareq} gives
\begin{equation} \label{eq:Aklchareq3d}
    (i\hbar)^3 + B_{12}B_{13}B_{23}(A^*_{12}A_{13}A^*_{23}-A_{12}A^*_{13}A_{23})+i\hbar(B_{12}^2\abs{A_{12}}^2+B_{13}^2\abs{A_{13}}^2+B_{23}^2\abs{A_{23}}^2) = 0 \, .
\end{equation}
The ansatz of having the off-diagonal elements of $\operator{A}$ be $A_{kl}=\hbar\exp(i\alpha_{kl})/\abs{B_{kl}} = A_{lk}^*$ when substituted above gives the additional condition that $\alpha_{12}-\alpha_{13}+\alpha_{23}=3\pi/2+2n\pi$ for integer $n$. The commutator has eigenvalues $i\hbar$, $i\hbar$, and $-2i\hbar$. One possible canonical pair is when $\alpha_{12}=\alpha_{13}=\alpha_{23}=3\pi/2 + 2n\pi$,
\begin{equation} \label{eq:3dcp}
    \operator{A}=\mqty(a_1 && \frac{i\hbar}{B_{12}} && \frac{i\hbar}{B_{13}} \\ -\frac{i\hbar}{B_{12}} && a_2 && \frac{i\hbar}{B_{23}} \\ -\frac{i\hbar}{B_{13}} && -\frac{i\hbar}{B_{23}} && a_3) \, , \qquad \operator{B}=\mqty(B_1 && 0 && 0 \\ 0 && B_2 && 0 \\ 0 && 0 && B_3) \, ,
\end{equation}
with the two-dimensional canonical domain
\begin{equation} \label{eq:3dcd}
    \cd = \qty{\ket{\varphi} = \phi_1\mqty(1 \\ 0 \\ -1) + \phi_2\mqty(0 \\ 1 \\ -1), \braket{\varphi|\varphi}=1, \phi_1,\phi_2 \in \mathbb{C}} \, .
\end{equation}
With the unitary evolution $\operator{U}(t)=\exp(-it\operator{B})$, we obtain the equivalent solutions for $t$ in the parameter invariant set $\tis = \qty{t = 2n\pi/\gcd(B_{12},B_{23}), n \in \mathbb{Z}}$, where $\gcd$ is the greatest common divisor. This gives us the class of solutions $\mathcal{C}(\operator{A},\operator{B},\cd;\operator{U}(t),\tis)$. One can obtain other equivalent solutions by another choice of unitary operator, though the parameter invariant set is always a set of measure zero (Theorem \ref{thm:tis}), i.e., for any unitary $\operator{U}(t)$, then the pair $(\operator{U}^\dagger(t)\operator{A}\operator{U}(t),\operator{U}^\dagger(t)\operator{B}\operator{U}(t))$ are canonical for $t$ in a set of measure zero. Once again we can obtain inequivalent solutions by appropriately varying the real constants $a_s$ in $\operator{A}$.

With the commutator eigenvalue $-2i\hbar$, we get an essentially canonical pair $(\operator{A},\operator{B})$ in the one-dimensional canonical domain containing $3^{-1/2}\smqty(1 \\ 1 \\ 1)$. This pair can be mapped to the canonical pair $(-\operator{A}/2,\operator{B})$ in this same one-dimensional canonical domain, which is inequivalent to the above canonical pair $(\operator{A},\operator{B})$ since their canonical domains have different dimensions.

Another way of generating inequivalent solutions presents itself for dimensions higher than two via projection (which are nonunitary operators). Consider the operator $\operator{P}_{12} = \op{1} + \op{2}$. We can have $\operator{A}_{12} = \operator{P}_{12}\operator{A}\operator{P}_{12} + \operator{F}$ (where $[\operator{F},\operator{B}]=0$) being a canonical pair with $\operator{B}$,
\begin{equation} \label{eq:3dcpproj}
    \operator{A}_{12}=\mqty(a_1 && \frac{i\hbar}{B_{12}} && 0 \\ -\frac{i\hbar}{B_{12}} && a_2 && 0 \\ 0 && 0 && a_3) \, , \qquad \operator{B}=\mqty(B_1 && 0 && 0 \\ 0 && B_2 && 0 \\ 0 && 0 && B_3) \, ,
\end{equation}
with the one-dimensional canonical domain
\begin{equation}
    \cd_{12} = \qty{\ket{\varphi}= \frac{1}{\sqrt{2}} \mqty(1 \\ -1 \\ 0)} \, ,
\end{equation}
giving us a solution $\mathcal{C}_{12}(\operator{A}_{12},\operator{B},\cd_{12})$ inequivalent with $\mathcal{C}(\operator{A},\operator{B},\cd)$. Here, we end up with a canonical pair that looks similar to our solution in two-dimensional Hilbert space. Therefore, there exists more ways to generate inequivalent solutions in higher dimensions (see Appendix \ref{app:3dnondegen} for a complete list of all possible solutions to the three-dimensional canonical commutation relation for nondegenerate $\operator{B}$). For $\operator{U}(t)=\exp(-it\operator{B})$, the parameter invariant set is $\tis=\qty{t = 2n\pi/B_{12}, n\in\mathbb{Z}}$, and thus, inequivalent solutions arising from projections may have different parameter invariant sets.

Next, we look at the case when $\operator{B}$ contains degeneracy. Suppose the eigenvalue $B_1$ has degeneracy $M_1 = 2$ and $B_2$ has degeneracy $M_2 = 1$ (nondegenerate). Let $\ket{1} = \ket{1,1}$, $\ket{2} = \ket{1,2}$ and $\ket{3}=\ket{2,1}$. For $\braket{k|\operator{A}|l} = A_{kl} = A_{lk}^*$, \eqref{eq:Askdegenchareq} gives
\begin{equation}
    (i\hbar)^3 + i\hbar(\abs{A_{13}}^2 + \abs{A_{23}}^2)B_{12}^2 = 0 \, .
\end{equation}
We look at one possible solution (see Appendix \ref{app:3ddegen} for a list of all possible solutions),
\begin{equation} \label{eq:3dcpdegen}
    \operator{A}=\mqty(a_1 && b &&\frac{i\hbar}{\sqrt{2}B_{12}} \\ b^* && a_2 && \frac{i\hbar}{\sqrt{2}B_{12}} \\ -\frac{i\hbar}{\sqrt{2}B_{12}} && -\frac{i\hbar}{\sqrt{2}B_{12}} && a_3) \, , \qquad \operator{B}=\mqty(B_1 && 0 && 0 \\ 0 && B_1 && 0 \\ 0 && 0 && B_2) \, ,
\end{equation}
which is a canonical pair in the one-dimensional canonical domain
\begin{equation} \label{eq:3dcddegen}
    \cd^+ = \qty{\ket{\varphi}= \frac{1}{2} \mqty(-1 \\ -1 \\ \sqrt{2}) } \, ,
\end{equation}
where $a_1$, $a_2$, and $a_3$ are real, and $b$ is complex. Note here that it is the term $\smqty(a_1 && b && 0 \\ b^* && a_2 && 0 \\ 0 && 0 && a_3)$ in $\operator{A}$ that commutes with $\operator{B}$. We also have the essentially canonical pair $(\operator{A},\operator{B})$ with respect to $-i\hbar$ with the canonical domain $\cd^- = \qty{2^{-1} \smqty( 1 \\ 1 \\ \sqrt{2})}$; this can be mapped to a canonical pair $(-\operator{A},\operator{B})$ with the same canonical domain $\cd^-$. The two solutions $\mathcal{C}(\operator{A},\operator{B},\cd^+)$ and $\mathcal{C}(-\operator{A},\operator{B},\cd^-)$ are equivalent with respect to the unitary operator $\operator{U}(\alpha) = \exp(-i\alpha\operator{B}/B_{12})$ with $\alpha = \pi$. Finally, given the unitary $\operator{U}(t)=\exp(-it\operator{B})$, then for $\ket{\varphi}\in\cd^\pm$, we have $\operator{U}(t)\ket{\varphi}\in\cd^\pm$ for $t$ in $\tis = \qty{t = 2n\pi/B_{12}, n \in \mathbb{Z}}$. We see that the parameter invariant set also depends on the degeneracy of $\operator{B}$.

\section{Canonical pairs in N-dimensional Hilbert space} \label{sec:nd}
We then proceed with the general case.
\begin{theorem} \label{thm:nddim}
    The canonical domain of a canonical pair in an $N$-dimensional Hilbert space is a closed subspace of the Hilbert space with dimension of at most $N-1$.
\end{theorem}
\begin{proof}
    From Theorem \ref{thm:degen}, the degeneracy of the commutator with eigenvalue $i\hbar$ is at most $N-1$, and thus, the dimension of its eigenspace is at most $N-1$. Since its eigenspace is the canonical domain, then it follows that the canonical domain is a closed subspace of the Hilbert space with dimension of at most $N-1$.
\end{proof}

\begin{remark}
    Solutions to the canonical commutation relation in finite-dimensional Hilbert space are always of closed-category, i.e., no dense-category solution exists.
\end{remark}

For nondegenerate $\operator{B}\ket{s}=B_s\ket{s}$, $s = 1, 2, \dotsc, N$, we can construct the canonical pair \cite{SPP-2024-PC-20}
\begin{equation} \label{eq:ndcpa}
    \operator{A} = \sum_{s=1}^N \sum_{\substack{s' = 1 \\ s' \neq s}}^N \frac{i\hbar}{B_s - B_{s'}} \ket{s}\bra{s'} + \sum_{s=1}^N a_s \ket{s}\bra{s} \, ,
\end{equation}
\begin{equation} \label{eq:ndcpb}
    \operator{B} = \sum_{s=1}^N B_s \ket{s}\bra{s} \, ,
\end{equation}
with the canonical domain
\begin{equation} \label{eq:ndcd}
    \cd = \qty{ \ket{\varphi} = \sum_{l=1}^{N-1} \sum_{k = l+1}^N c_{kl} (\ket{k} - \ket{l}), c_{kl} \in \mathbb{C}} \, .
\end{equation}
Suppose $\operator{U}(t) = \exp(-i\operator{B}t)$, and let $\delta(\operator{B})$ be the set of all $B_k - B_l$ (for all $k > l$), i.e., the set of all absolute differences of the eigenvalues of $\operator{B}$. The parameter invariant set is then
\begin{equation}
    \tis = \qty{ t = \frac{2n\pi}{\gcd(\delta(\operator{B}))} , n \in \mathbb{Z} }
\end{equation}
where $\gcd(\delta(\operator{B}))$ is the greatest common divisor of the set $\delta(\operator{B})$. We see that the parameter invariant set has total measure zero, and depends on the nature of the eigenvalues of $\operator{B}$. Equivalent solutions and inequivalent solutions arising from varying $a_s$ do not affect the parameter invariant set. We then have the class of solutions $\mathcal{C}_N(\operator{A},\operator{B},\cd;\exp(-i\operator{B}t),\tis)$ in $N$-dimensional Hilbert space.

\begin{proposition} \label{prop:dim}
    The subspace \eqref{eq:ndcd} has dimension of $N-1$.
\end{proposition}
\begin{proof}
    The $l=1$ term in \eqref{eq:ndcd} has $N-1$ linearly independent vectors: $\ket{2} - \ket{1}, \ket{3} - \ket{1}, \dotsc, \ket{N-1} - \ket{1}, \ket{N} - \ket{1}$. For all other $l > 1$ with $k > l$ terms, note that any $\ket{k} - \ket{l}$ can be written as $(\ket{k} - \ket{1}) - (\ket{l} - \ket{1})$, and are thus linearly dependent. Therefore, the subspace \eqref{eq:ndcd} has dimension $N - 1$.
\end{proof}

\begin{remark}
    From Theorem \ref{thm:nddim}, we know that the largest possible canonical domain has dimension $N-1$. With Proposition \ref{prop:dim}, we see that the canonical pair given by \eqref{eq:ndcpa} and \eqref{eq:ndcpb} provides a way of constructing canonical pairs with the largest possible canonical domain.
\end{remark}

Notice the similarity with the characteristic (time) operator in \textit{infinite}-dimensional Hilbert space \cite{galapon2002self}
\begin{equation}
    \operator{A} = \sum_{s=1}^\infty \sum_{\substack{s' = 1 \\ s' \neq s}}^\infty \frac{i\hbar}{B_s - B_{s'}} \ket{s}\bra{s'} + \sum_{s=1}^\infty a_s \ket{s}\bra{s} \, ,
\end{equation}
\begin{equation}
    \operator{B} = \sum_{s=1}^\infty B_s \ket{s}\bra{s} \, ,
\end{equation}
\begin{equation}
    \cd = \qty{ \ket{\varphi} = \sum_{l=1}^{M-1} \sum_{k = l+1}^M c_{kl} (\ket{k} - \ket{l}), c_{kl} \in \mathbb{C}, M < \infty} \, ,
\end{equation}
wherein $\cd$ is dense, giving the dense-category solution $\mathcal{C}_\infty(\operator{A},\operator{B},\cd)$ for $\sum_{s=1}^\infty B_s^{-1} < \infty$. Suppose once again that $\operator{U}(t) = \exp(-i\operator{B}t)$ and $\delta(\operator{B})$ be the set of all $B_k - B_l$. The parameter invariant set is similarly \cite{farrales2025characteristic}
\begin{equation}
    \tis = \qty{ t = \frac{2n\pi}{\gcd(\delta(\operator{B}))} , n \in \mathbb{Z} } \, ,
\end{equation}
giving the class of solutions $\mathcal{C}_\infty(\operator{A},\operator{B},\cd;\exp(-i\operator{B}t),\tis)$. The class of closed-category solutions $\mathcal{C}_N$ in $N$-dimensional Hilbert space is just the projection of the infinite-dimensional solution $\mathcal{C}_\infty$ to a subspace of the infinite-dimensional Hilbert space spanned by $\ket{1}, \ket{2}, \dotsc, \ket{N-1}, \ket{N}$. In other words, the projection of the characteristic operator to an $N$-dimensional subspace corresponds to a canonical pair with the largest possible canonical domain in that subspace. However, it does not necessarily work the other way around: letting $N \to \infty$ for the $N$-dimensional canonical solution does not necessarily give a canonical pair in infinite-dimensional Hilbert space as the condition on $B_s$ may not be satisfied.

Consider a projection of $\operator{A}$ to the two-dimensional subspace spanned by $\ket{k}$ and $\ket{l}$. We have $\operator{A} = (i\hbar/B_{kl})(\ket{k}\bra{l} - \ket{l}\bra{k})$ and $\operator{B} = \sum_{s=1}^N B_s \ket{s}\bra{s}$ being an essentially canonical pair in the one-dimensional canonical domain $\cd^\pm = \qty{2^{-1/2} (\ket{k} \mp \ket{l}) }$, i.e., $[\operator{A},\operator{B}]\ket{\varphi_\pm} = \pm i\hbar \ket{\varphi_\pm}$ for $\ket{\varphi_\pm}\in\cd^\pm$. Firstly, we have the canonical solution $\mathcal{C}(\operator{A},\operator{B},\cd^+)$. Secondly, the essentially canonical solution with respect to $-i\hbar$ can be mapped to the canonical solution $\mathcal{C}(-\operator{A},\operator{B},\cd^-)$, and is an equivalent solution to $\mathcal{C}(\operator{A},\operator{B},\cd^+)$ with respect to $\operator{U} = \exp(i\pi \operator{B}/B_{kl})$. Note that the Hilbert space minus the two canonical domains $\cd^\pm$ can be written as $\hilbertspace \setminus(\cd^+ \cup \cd^-)=\qty{\ket{\varphi_\perp} = \sum_{\substack{s = 1 \\ s \neq k,l}}^N \varphi_s \ket{s},\varphi_s \in \mathbb{C}}$, and with the commutator explicitly written as $[\operator{A},\operator{B}] = -i\hbar(\ket{k}\bra{l} + \ket{l}\bra{k})$, we then see that $[\operator{A},\operator{B}]\ket{\varphi_\perp} = 0$. The commutator thus only has two nonzero nondegenerate eigenvalues, $\pm i\hbar$, which corresponds to two equivalent solutions (when one is mapped to the essentially canonical relation of the other) with one-dimensional canonical domains.

The solutions for the 3-dimensional case in Appendix \ref{app:3dsoln} suggest that the general solution in the $N$-dimensional case has the form
\begin{equation} \label{eq:ndcpagen}
    \operator{A} = \sum_{s=1}^N \sum_{\substack{s' = 1 \\ s' \neq s}}^N \beta_{s,s'}\frac{i\hbar}{B_s - B_{s'}} e^{i\alpha_{s,s'}} \ket{s}\bra{s'} + \sum_{s=1}^N a_s \ket{s}\bra{s} \, ,
\end{equation}
which forms a canonical pair with the nondegenerate $\operator{B}$ for real $\alpha_{s,s'}$ and complex $\beta_{s,s'}$ where $\alpha_{s,s'} = -\alpha_{s',s}$ and $\beta_{s,s'} = \beta_{s',s}^*$, where $\alpha_{s,s'}$ and $\beta_{s,s'}$ should be such that \eqref{eq:Aklchareq} is satisfied. The form of the canonical domain would then depend on $\alpha_{s,s'}$ and $\beta_{s,s'}$,
\begin{equation}
    \cd = \qty{\ket{\varphi} = \sum_{s=1}^N \varphi_s \ket{s}, \sum_{\substack{s' = 1 \\ s' \neq s}}^N \beta_{s,s'} e^{i\alpha_{s,s'}} \varphi_{s'} = -\varphi_s,\varphi_s\in\mathbb{C}} \, .
\end{equation}

Similarly, for degenerate $\operator{B}$, we have the general form
\begin{align} \label{eq:ndcpadegengen}
    \operator{A} &= \sum_{s=1}^L \sum_{\substack{s' = 1 \\ s' \neq s}}^L \sum_{r=1}^{M_s} \sum_{r' = 1}^{M_{s'}} \beta_{s,r,s',r'}\frac{i\hbar}{B_s - B_{s'}} e^{i\alpha_{s,r,s',r'}} \ket{s,r}\bra{s',r'} \nonumber \\
    &\quad + \sum_{s=1}^L \sum_{r=1}^{M_s} a_{s,r} \ket{s,r}\bra{s,r} + \sum_{s=1}^L \sum_{r=1}^{M_s} \sum_{\substack{r' = 1 \\ r' \neq r}}^{M_s} b_{s,r,r'} \ket{s,r}\bra{s,r'} \, ,
\end{align}
forming a canonical pair with $\operator{B} = \sum_{s=1}^L \sum_{r=1}^{M_s} B_s \ket{s,r}\bra{s,r}$ for values of complex $\beta_{s,r,s',r'}=\beta_{s',r',s,r}^*$ and $b_{s,r,r'}=b_{s,r',r}^*$, and real $\alpha_{s,r,s',r'}=-\alpha_{s',r',s,r}$ such that \eqref{eq:Askdegenchareq} is satisfied, with the canonical domain
\begin{equation} \label{eq:ndcddegengen}
    \cd = \qty{ \ket{\varphi} = \sum_{s=1}^L \sum_{r=1}^{M_s} \varphi_{s,r} \ket{s,r}, \sum_{\substack{s' = 1 \\ s' \neq s}}^L \sum_{r'=1}^{M_{s'}} \beta_{s,r,s',r'} e^{i\alpha_{s,r,s',r'}}\varphi_{s',r'}=\varphi_{s,r},\varphi_{s,r}\in\mathbb{C}} \, .
\end{equation}
This looks similar to the result in \cite{galapon2002self,SPP-2024-PC-20} but here we considered the general case wherein the degeneracy is not constant (i.e., $M_s$ is not necessarily the same for all $s$). We recover the nondegenerate solution if $M_s = 1$ for all $s$. Note that when $N=\sum_{s=1}^L M_s$ approaches infinity and assuming $\sum_s B_s^{-1} < \infty$, we also obtain the infinite-dimensional result with $\operator{B}$ containing degeneracy, and thus, the finite-dimensional result is once again seen to be its projection to a smaller subspace.

Essentially, the method of obtaining all possible solutions starts with $\operator{A}$ in \eqref{eq:ndcpadegengen}: adjust $\alpha_{s,r,s',r'}$ to obtain equivalent solutions, vary $a_{s,r}$ and $b_{s,r,r'}$ to adjust the term commuting with $\operator{B}$ which gives either equivalent or inequivalent solutions, or adjust $\beta_{s,r,s',r'}$ which could be thought of as the coefficients of the sum of all possible 2D projections of $\operator{A}$ to obtain inequivalent solutions. That is, if $(\operator{A},\operator{B})$ is a canonical pair, then so is $(\sum_{s,r,s',r'} \beta_{s,r,s',r'} \operator{P}_{s,r,s',r'}\operator{A}\operator{P}_{s,r,s',r'},\operator{B})$, where $\operator{P}_{s,r,s',r'} = \op{s,r} + \op{s',r'}$ and for some appropriate values of $\beta_{s,r,s',r'}$. 

Equivalent solutions share the same $\tis$. Equivalent and inequivalent solutions arising from adding a term that commutes with $\operator{B}$ also have the same $\tis$. Meanwhile, inequivalent solutions obtained via projection may have a different parameter invariant set.

\begin{proposition} \label{prop:tis}
    Consider the class of solutions $\mathcal{C}(\operator{A},\operator{B},\cd;\exp(-i\operator{B}t),\tis)$. Let $\operator{B}\ket{s} = B_s\ket{s}$ and $\operator{P}_s = \ket{s}\bra{s}$. The parameter invariant set $\tis$ is then
    \begin{equation}
        \tis = \qty{ t = \frac{2n\pi}{\gcd( \tilde{\delta}(\operator{B}))}, n \in \mathbb{Z} }
    \end{equation}
    where $\gcd$ is the greatest common divisor, and $\tilde{\delta}(\operator{B})$ is the set of all eigenvalue differences $B_k - B_l$ not including the differences that contain $B_s$ if $\operator{P}_s\ket{\varphi} = 0$ for all $\ket{\varphi}\in\cd$.
\end{proposition}
\begin{proof}
    We can write $\ket{\varphi} = \sum_j \alpha_j \ket{j} \in \cd$. If $\alpha_j \neq 0$ for all $j$, then $\exp(-i\operator{B}t)\ket{\varphi} = \sum_j \alpha_j \exp(-iB_j t) \ket{j} \in \cd$ for $t \in \tis = \qty{t = 2n\pi/\gcd(\delta(\operator{B})),n \in \mathbb{Z}}$ where $\delta(\operator{B})$ is the set of all $B_k - B_l$ ($k > l$). If at least one $\alpha_j$ is zero for all $\ket{\varphi}\in \cd$, i.e., for some $s$, $\operator{P}_s\ket{\varphi}=0$ for all $\ket{\varphi}\in\cd$, then $t \in \tis = \qty{t = 2n\pi/\gcd(\tilde{\delta}(\operator{B})),n\in\mathbb{Z}}$ where $\tilde{\delta}(\operator{B})$ is the set of all $B_k - B_l$ not including differences containing $B_s$.
\end{proof}

\noindent
Thus, inequivalent solutions arising from projections to a smaller subspace will have $\tilde{\delta}(\operator{B})$ not include differences containing $B_s$ if $\ket{s}$ is outside that smaller subspace.

\begin{remark}
    If there are two elements of $\tilde{\delta}(\operator{B})$ that are incommensurate (i.e., if there exists a $(B_k-B_l)/(B_{k'}-B_{l'})$ that is irrational), then the parameter invariant set is $\tis=\qty{t=0}$. Otherwise, it is countably infinite. For all cases, the parameter invariant set has total measure zero (Theorem \ref{thm:tis}).
\end{remark}

Next, we shall generalize Proposition \ref{prop:2dtracecomm} and recover the well-known theorem that every traceless matrix is a commutator \cite{putnam1967commutation,shoda1936einige}.

\begin{lemma} \label{lem:unitary}
    Given a traceless matrix $\operator{C}$, there exists a unitary matrix $\operator{U}$ such that the diagonal of $\operator{U}^\dagger\operator{C}\operator{U}$ consists only of zeros. 
\end{lemma}
\begin{proof}
    Let $\operator{C} \ket{\varphi_s} = c_s \ket{\varphi_s}$, with $\sum_{s=1}^N c_s = 0$ since $\operator{C}$ is traceless. Consider the operator $\operator{U}$ with matrix elements $\braket{\varphi_s|\operator{U}|\varphi_{s'}} = N^{-1/2} \exp(i (s-1) (s'-1) 2\pi/N)$. We see that $\operator{U}$ is unitary:
    \begin{equation}
        \braket{\varphi_s|\operator{U}\operator{U}^\dagger|\varphi_{s'}} = \sum_{k=1}^N \braket{\varphi_s|\operator{U}|\varphi_k}\braket{\varphi_k|\operator{U}^\dagger|\varphi_{s'}} = \frac{1}{N} \sum_{k=1}^N e^{i (k-1) (s-s') \frac{2\pi}{N}} = \delta_{ss'} = \braket{\varphi_s|\operator{U}^\dagger\operator{U}|\varphi_{s'}} \, ,
    \end{equation}
    where $\delta_{ss'}$ is the Kronecker delta. Next, we see that the diagonal of $\operator{U}^\dagger\operator{C}\operator{U}$ consists only of zeros:
    \begin{equation}
        \braket{\varphi_s|\operator{U}^\dagger\operator{C}\operator{U}|\varphi_s} = \frac{1}{N} \sum_{k=1}^N c_k = 0 \, .
    \end{equation}
\end{proof}

\begin{theorem} \label{thm:commN}
    Every nontrivial $N \times N$ traceless matrix is a commutator.
\end{theorem}
\begin{proof}
    Consider the traceless $\operator{C} = \sum_{s=1}^N c_s \op{\varphi_s}$ where $\sum_{s=1}^N c_s = 0$, at least two $c_s$ are nonzero, and $\braket{\varphi_s|\varphi_k}=\delta_{sk}$. From Lemma \ref{lem:unitary}, there exists a unitary operator $\operator{U}$ such that in the new basis $\ket{s} = \operator{U}\ket{\varphi_s}$, the diagonal of $\operator{C}$ consists only of zeros.

    Next, we construct the operator $\operator{B} = \sum_{s=1}^N B_s \op{s}$ where $B_s$ are constants. We then have the operator
    \begin{equation}
        \operator{A} = \sum_{s=1}^N \sum_{\substack{s' = 1 \\ s' \neq s}}^N \frac{\braket{s|\operator{C}|s'}}{B_{s'} - B_s} \ket{s}\bra{s'} + \sum_{s=1}^N a_s \ket{s}\bra{s} \, ,
    \end{equation}
    where $a_s$ are constants, such that $\operator{U}^\dagger\operator{C}\operator{U} = [\operator{A},\operator{B}]$. Therefore, $\operator{C}=[\operator{U}\operator{A}\operator{U}^\dagger,\operator{U}\operator{B}\operator{U}^\dagger]$ is a commutator with $\operator{C}\ket{\varphi_s} = c_s\ket{\varphi_s}$.
\end{proof}

\begin{proposition} \label{prop:everyN}
    To every element $\ket{\varphi}$ of the $N$-dimensional Hilbert space $\hilbertspace$, there exists a commutator $\operator{C}$ such that $\operator{C}\ket{\varphi}=c\ket{\varphi}$, where $c$ is a nonzero constant.
\end{proposition}
\begin{proof}
    For a given $\ket{\varphi}$, consider $\operator{C} = \sum_{s=1}^N c_s \op{\varphi_s}$ where $\sum_{s=1}^N c_s = 0$, at least two $c_s$ are nonzero, and $\braket{\varphi_s|\varphi_k}=\delta_{sk}$, such that (for $s = r$) $c_r = c \neq 0$ and $\ket{\varphi_r} = \ket{\varphi}$. Thus, $\operator{C}\ket{\varphi} = c \ket{\varphi}$. From Theorem \ref{thm:commN}, $\operator{C}$ is a commutator.
\end{proof}

\section{Uncertainty relation} \label{sec:uncertainty}
The uncertainty relation between a pair of operators is related to the commutation relation they satisfy \cite{hall2013quantum}.

\begin{definition}
    The uncertainty of an operator $\operator{A}$ with respect to the state $\ket{\varphi} \in \hilbertspace$ is given by
    \begin{equation} \label{eq:unc}
        \Delta_\varphi \operator{A} = \sqrt{\braket{\varphi|(\operator{A} - \braket{\varphi|\operator{A}|\varphi}\operator{I}_\hilbertspace)^2|\varphi}} = \sqrt{\braket{\varphi|\operator{A}^2|\varphi} - \braket{\varphi|\operator{A}|\varphi}^2} \, ,
    \end{equation}
    where $\operator{I}_\hilbertspace$ is the identity in $\hilbertspace$.
\end{definition}

\begin{lemma} \label{lem:unczero}
    $\Delta_\varphi \operator{A} = 0$ if and only if $\ket{\varphi}$ is an eigenket of $\operator{A}$.
\end{lemma}
\begin{proof}
    If $\Delta_\varphi \operator{A} = 0$, then from \eqref{eq:unc}, $\braket{\varphi|(\operator{A} - \braket{\varphi|\operator{A}|\varphi}\operator{I}_\hilbertspace)^2|\varphi} = 0$, implying $\operator{A}\ket{\varphi}= \braket{\varphi|\operator{A}|\varphi}\ket{\varphi}$, i.e., $\ket{\varphi}$ is an eigenket of $\operator{A}$ with eigenvalue $\braket{\varphi|\operator{A}|\varphi}$. Conversely, if $\operator{A}\ket{\varphi} = a \ket{\varphi}$, then $\Delta_\varphi \operator{A} = 0$.
\end{proof}

\begin{proposition}
    Given an essentially canonical pair $(\operator{A},\operator{B})$ with the canonical domain $\cd$, then
    \begin{equation}
        (\Delta_\varphi \operator{A}) (\Delta_\varphi \operator{B}) \neq 0 \, ,
    \end{equation}
    for $\ket{\varphi} \in \cd$.
\end{proposition}
\begin{proof}
    From Theorem \ref{thm:noeigenket}, the eigenkets of $\operator{A}$ and $\operator{B}$ are not in $\cd$. Then, from Lemma \ref{lem:unczero}, $\Delta_\varphi \operator{A} \neq 0$ for $\ket{\varphi} \in \cd$. Similarly, $\Delta_\varphi \operator{B} \neq 0$ for $\ket{\varphi} \in \cd$.
\end{proof}

\noindent
We then see that for any essentially canonical pair of observables, their uncertainties (with respect to the states in the canonical domain) are never zero. The uncertainties can only be zero at their eigenkets, which a state can only pass through after it evolves and moves outside the canonical domain. In this scenario, the pair fails to be canonical and the uncertainties can be zero without any issue. Thus, for states in the canonical domain, this \textit{uncertainty relation} is always greater than or equal to a nonzero real constant. We have the following theorem \cite{hall2013quantum}, generalized for essentially canonical pairs.

\begin{theorem}
    Given a Hermitian essentially canonical pair $(\operator{A},\operator{B})$ in the canonical domain $\cd$, i.e., $[\operator{A},\operator{B}]\ket{\varphi} = ic \ket{\varphi}$ for $\ket{\varphi} \in \cd$ and nonzero real $c$, then the pair satisfies
    \begin{equation} \label{eq:essuncrel}
        (\Delta_\varphi \operator{A}) (\Delta_\varphi \operator{B}) \geq \frac{\abs{c}}{2} \, ,
    \end{equation}
    for all $\ket{\varphi} \in \cd$.
\end{theorem}
\begin{proof}
    Consider the operator $\operator{A}' = \operator{A} - \braket{\varphi|\operator{A}|\varphi}\operator{I}_\hilbertspace$ and $\operator{B}' = \operator{B} - \braket{\varphi|\operator{B}|\varphi}\operator{I}_\hilbertspace$, where $\operator{I}_\hilbertspace$ is the identity in $\hilbertspace$. Since $\operator{A}$ and $\operator{B}$ are Hermitian, so are $\operator{A}'$ and $\operator{B}'$, and thus, $\braket{\varphi|(\operator{A}')^2|\varphi} = \braket{\operator{A}'\varphi|\operator{A}'\varphi}$ and $\braket{\varphi|(\operator{B}')^2|\varphi} = \braket{\operator{B}'\varphi|\operator{B}'\varphi}$. The Cauchy-Schwarz inequality gives
    \begin{align}
        \braket{\varphi|(\operator{A}')^2|\varphi} \braket{\varphi|(\operator{B}')^2|\varphi} &\geq \abs{\braket{\operator{A}'\varphi|\operator{B}'\varphi}}^2 \label{eq:cs1} \\
        &\geq \abs{\Im \braket{\operator{A}'\varphi|\operator{B}'\varphi}}^2 \label{eq:cs2} \\
        &= \frac{1}{4} \abs{\braket{\operator{A}'\varphi|\operator{B}'\varphi} - \braket{\operator{B}'\varphi|\operator{A}'\varphi}}^2 \\
        &= \frac{1}{4} \abs{\braket{\varphi|[\operator{A}',\operator{B}']\varphi}}^2 \, .
    \end{align}
    Note that $[\operator{A}',\operator{B}']\ket{\varphi} = [\operator{A},\operator{B}]\ket{\varphi} = ic\ket{\varphi}$ for $\ket{\varphi} \in \cd$. Thus,
    \begin{equation}
        \braket{\varphi|(\operator{A}')^2|\varphi} \braket{\varphi|(\operator{B}')^2|\varphi} \geq \frac{\abs{c}^2}{4} \, .
    \end{equation}
    Since $\Delta_\varphi \operator{A} = \sqrt{\braket{\varphi|(\operator{A} - \braket{\varphi|\operator{A}|\varphi}\operator{I}_\hilbertspace)^2|\varphi}} = \sqrt{\braket{\varphi|(\operator{A}')^2|\varphi}}$ and $\Delta_\varphi \operator{B} = \sqrt{\braket{\varphi|(\operator{B} - \braket{\varphi|\operator{B}|\varphi}\operator{I}_\hilbertspace)^2|\varphi}} = \sqrt{\braket{\varphi|(\operator{B}')^2|\varphi}}$, we therefore have
    \begin{equation}
        (\Delta_\varphi \operator{A})(\Delta_\varphi \operator{B}) \geq \frac{\abs{c}}{2} \, .
    \end{equation}
\end{proof}

\noindent 
The inequality \eqref{eq:essuncrel} is called the uncertainty relation of $\operator{A}$ and $\operator{B}$. We will use the following proposition to identify the existence of minimum uncertainty states \cite{hall2013quantum}.
\begin{proposition} \label{prop:minunc}
    The uncertainty relation \eqref{eq:essuncrel} becomes an equality if and only if $\ket{\varphi}$ is an eigenket of $\operator{A} - i\gamma\operator{B}$ for some nonzero real $\gamma$.
\end{proposition}
\begin{proof} 
    $(\Delta_\varphi \operator{A})(\Delta_\varphi \operator{B}) = \abs{c}/2$ if and only if \eqref{eq:cs1} and \eqref{eq:cs2} in the proof above are equalities, i.e.,
    \begin{equation}
        \braket{\varphi|(\operator{A}')^2|\varphi} \braket{\varphi|(\operator{B}')^2|\varphi} = \abs{\braket{\operator{A}'\varphi|\operator{B}'\varphi}}^2 = \abs{\Im \braket{\operator{A}'\varphi|\operator{B}'\varphi}}^2 \, .
    \end{equation}
    Thus, the equalities hold if and only if $\operator{A}'\ket{\varphi} = 0$, or $\operator{B}'\ket{\varphi} = 0$, or $\operator{A}'\ket{\varphi}=i\gamma\operator{B}'\ket{\varphi}$ (nonzero real $\gamma$). Since $\ket{\varphi} \in \cd$ cannot be an eigenket of $\operator{A}$ or $\operator{B}$ (and thus, are also not eigenkets of $\operator{A}'$ and $\operator{B}'$), then we are only left with having $\ket{\varphi}$ be an eigenket of $\operator{A}' - i\gamma\operator{B}'$ with eigenvalue zero. From the definition $\operator{A}' = \operator{A} - \braket{\varphi|\operator{A}|\varphi}\operator{I}_\hilbertspace$ and $\operator{B}' = \operator{B} - \braket{\varphi|\operator{B}|\varphi}\operator{I}_\hilbertspace$, we then have $(\operator{A}-i\gamma\operator{B})\ket{\varphi} = (\braket{\varphi|\operator{A}|\varphi} - i\gamma \braket{\varphi|\operator{B}|\varphi})\ket{\varphi}$, i.e., $\ket{\varphi}$ is an eigenket of $\operator{A}-i\gamma\operator{B}$.
\end{proof}

As an example, consider the 2-dimensional canonical pair \eqref{eq:2dcp} with their canonical domain \eqref{eq:2dcd}. This pair satisfies the uncertainty relation \eqref{eq:essuncrel} where $c = \hbar$. We see that $\ket{\varphi}\in\cd$ is an eigenket of $\operator{A}-i\gamma\operator{B}$ if $\gamma = -2\hbar/B_{12}^2 - (a_2 - a_1)/iB_{12}$. From Proposition \ref{prop:minunc}, the uncertainty relation is saturated (i.e., it becomes an equality with $\hbar/2$) only when $\gamma$ is purely real and nonzero. For the two-dimensional case, this only happens when $a_1 = a_2$. That is, if the term in $\operator{A}$ which commutes with $\operator{B}$ is proportional to the identity, then the uncertainty relation is $(\Delta_\varphi \operator{A}) (\Delta_\varphi \operator{B}) = \hbar/2$. In \cite{farrales2025characteristic}, we similarly see that the uncertainty relation is saturated for the two-dimensional projection of the characteristic time operator. We then observe one possible difference between the inequivalent solutions in 2D: adjusting $a_1$ and $a_2$ (e.g., to some $a_1\neq a_2$) affects the uncertainty relation between the canonical pair. 

Consider next the 3D case wherein $\operator{B}$ is nondegenerate. For the first canonical pair example \eqref{eq:3dcp}, we see that no such real nonzero $\gamma$ exists such that any $\ket{\varphi}\in\cd$ is an eigenket of $\operator{A} - i\gamma\operator{B}$, and thus, the uncertainty relation cannot be saturated. For the second example \eqref{eq:3dcpproj}, i.e., after performing a two-dimensional projection of \eqref{eq:3dcp}, we see that for $a_1 = a_2$, we find a real nonzero $\gamma = 2\hbar/B_{12}^2$ such that $\ket{\varphi}$ in its canonical domain is an eigenket of $\operator{A}-i\gamma\operator{B}$. Consider next when $\operator{B}$ contains degeneracy, i.e., we have the canonical pair \eqref{eq:3dcpdegen} with canonical domain \eqref{eq:3dcddegen}. When $a_1 = a_2 = a_3$ and $b = 0$, there exists a real $\gamma = 2\hbar/B_{12}^2$ such that $\ket{\varphi}$ is an eigenket of $\operator{A}-i\gamma\operator{B}$. The uncertainty relation can then be saturated for the latter two cases.

In general, for nondegenerate $\operator{B}$ forming a canonical pair with \eqref{eq:ndcpa}, it would seem that no minimum uncertainty state exists when the dimension $N$ is greater than two, since
\begin{align}
    (\operator{A}-i\gamma\operator{B})\ket{\varphi} = \sum_{l=1}^N \sum_{k=l+1}^N c_{k,l} \Bigg( &\qty(-\frac{i\hbar}{B_k-B_l} + a_k - i\gamma B_k)\ket{k} \nonumber \\
    &- \qty(\frac{i\hbar}{B_k-B_l}+a_l-i\gamma B_l)\ket{l} \nonumber \\
    &+ i\hbar(B_k - B_l) \sum_{s\neq k,l}^N \frac{1}{(B_s-B_k)(B_s-B_l)}\ket{s} \Bigg) \, .
\end{align}
And so, $\ket{\varphi}$ can become an eigenket of $\operator{A}-i\gamma\operator{B}$ if $a_k = a_l$ and $B_k - B_l$ is constant for all $k > l$, but this is only possible when $N = 2$.

Consider next a degenerate $\operator{B}$ with only two distinct eigenvalues, $B_1$ and $B_2$. With $L=2$ in \eqref{eq:ndcpadegengen}, this forms a canonical pair with
\begin{align} \label{eq:ndcpadegen2}
    \operator{A} &= \frac{1}{\sqrt{M_1 M_2}} \sum_{r=1}^{M_1} \sum_{r'=1}^{M_2} \frac{i\hbar}{B_1 - B_2} \ket{1,r}\bra{2,r'} + \frac{1}{\sqrt{M_1 M_2}} \sum_{r=1}^{M_2} \sum_{r'=1}^{M_1} \frac{i\hbar}{B_2-B_1}\ket{2,r}\bra{1,r'} \nonumber \\
    &\quad + \sum_{r=1}^{M_1} a_{1,r} \op{1,r} + \sum_{r=1}^{M_2} a_{2,r} \op{2,r} + \sum_{r=1}^{M_1} \sum_{\substack{r' = 1 \\ r' \neq r}}^{M_1} b_{1,r,r'} \ket{1,r}\bra{1,r'} + \sum_{r=1}^{M_2} \sum_{\substack{r' = 1 \\ r' \neq r}}^{M_2} b_{2,r,r'} \ket{2,r}\bra{2,r'} \, ,
\end{align}
where we let $\beta_{s,r,s',r'} = 1/\sqrt{M_1 M_2}$ so that \eqref{eq:Askdegenchareq} is satisfied. From \eqref{eq:ndcddegengen}, this pair has the one-dimensional canonical domain
\begin{equation}
    \cd = \qty{\ket{\varphi} = -\frac{1}{\sqrt{2 M_1}} \sum_{r=1}^{M_1} \ket{1,r} + \frac{1}{\sqrt{2 M_2}} \sum_{r=1}^{M_2} \ket{2,r} } \, .
\end{equation}
Then, $\ket{\varphi}\in\cd$ is an eigenket of $\operator{A}-i\gamma\operator{B}$ if $a_{1,r}=a_{2,r}$ and $b_{1,r,r'}=b_{2,r,r'}=0$ for all $r,r'$, with $\gamma = 2\hbar/B_{12}^2$.

Therefore, in $N$-dimensional Hilbert space, we see that under certain conditions on the term that commutes with $\operator{B}$, we can have a canonical pair $(\operator{A},\operator{B})$ satisfy $\Delta_\varphi\operator{A} \Delta_\varphi\operator{B}= \hbar/2$ for $\ket{\varphi}\in\cd$ if any of the following is true: (a) the Hilbert space is two-dimensional; (b) $\operator{B}$ is degenerate and only has two distinct eigenvalues; (c) $\operator{A}$ is projected to a two-dimensional subspace; or (d) $B$ is degenerate and $\operator{A}$ is projected to the subspace with only two possible eigenvalues of $\operator{B}$. In other words, if $\operator{A}$ only depends on two eigenvalues of $\operator{B}$, then a minimum uncertainty is possible.

As a quick remark, note that \eqref{eq:ndcpadegen2} and $\operator{B} = \sum_{s=1}^L \sum_{r=1}^{M_s} B_s \ket{s,r}\bra{s,r}$ form an essentially canonical pair with canonical domain $\cd_\pm = \qty{\ket{\varphi} = -\mp (2 M_1)^{-1/2} \sum_{r=1}^{M_1} \ket{1,r} + (2 M_2)^{-1/2} \sum_{r=1}^{M_2} \ket{2,r} }$, i.e., $[\operator{A},\operator{B}]\ket{\varphi_\pm} = \pm i\hbar \ket{\varphi_\pm}$ for $\ket{\varphi_\pm} \in \cd_\pm$. We see that the space outside the canonical domains is $\hilbertspace \setminus (\cd_+ \cup \cd_-)=\qty{\ket{\varphi_\perp} = \sum_{k=1}^{M_1 - 1} \alpha_k( -\ket{1,1} + \ket{1,k+1} ) + \sum_{k=1}^{M_2 - 2} \beta_k ( -\ket{2,1} + \ket{2,k+1} ) }$. We then have $[\operator{A},\operator{B}]\ket{\varphi_\perp} = 0$, and thus, we once again have a commutator with only two nonzero nondegenerate eigenvalues, $\pm i\hbar$. There are then two ways to have zero eigenvalues in the commutator (and thus, decreasing the maximum possible dimension of the canonical domains): either through projection or through degeneracy.

Finally, note that equivalent solutions satisfy the same uncertainty relation, i.e., unitary transformations do not affect the uncertainty relation. If for a particular solution the uncertainty relation is an equality, then all its equivalent solutions satisfy the same property. Consider $\mathcal{C}_1(\operator{A}_1,\operator{B}_1,\cd_1)$ and $\mathcal{C}_2(\operator{A}_2,\operator{B}_2,\cd_2)$ being equivalent solutions under some unitary $\operator{U}$. Then, the commutators $\operator{C}_1=[\operator{A}_1,\operator{B}_1]$ and $\operator{C}_2=[\operator{A}_2,\operator{B}_2]$ are also unitarily equivalent, with $\cd_1$ and $\cd_2$ having the same dimensions. With $\operator{C}_1 \ket{\varphi_1} = i\hbar \ket{\varphi_1}$ for $\ket{\varphi_1}\in\cd_1$, then $\operator{C}_2 \ket{\varphi_2} = i\hbar \ket{\varphi_2}$ for $\ket{\varphi_2} = \operator{U}^\dagger\ket{\varphi_1} \in \cd_2$. Similarly, $\operator{A}_1 - i\gamma\operator{B}_1$ is unitarily equivalent to $\operator{A}_2 - i\gamma\operator{B}_2$. If $\ket{\varphi_1}\in\cd_1$ is an eigenket of $\operator{A}_1 - i\gamma\operator{B}_1$ for some real constant $\gamma$, then $\ket{\varphi_2}\in\cd_2$ is an eigenket of $\operator{A}_2-i\gamma\operator{B}_2$ with the same $\gamma$. Thus, the saturation of the uncertainty relation is preserved under unitary transformations, and is a property that is shared with the equivalent solutions. In contrast, inequivalent solutions affect whether the uncertainty relation is saturated or not.

\section{Time operators} \label{sec:time}
Due to historical reasons, we shall call any operator that forms a canonical pair with the Hamiltonian as a time operator. This definition is justified since, as we will show, the measurement of the time operator can help us access the parametric time. Aside from that, we shall also call operators forming an essentially canonical pair (with respect to $-i\hbar$) with the Hamiltonian as a time operator.
\begin{definition}
    Given the Hamiltonian $\operator{H}$, a Hermitian operator $\operator{T}$ is called a time operator if $(\operator{T},\operator{H})$ or $(\operator{H},\operator{T})$ is a canonical pair in its corresponding canonical domain. The time operator satisfies the time-energy canonical commutation relation
    \begin{equation}
        [\operator{T},\operator{H}]\ket{\varphi}=\pm i\hbar\ket{\varphi} ,
    \end{equation}
    for all $\ket{\varphi}$ in the canonical domain $\cd^\pm$.
\end{definition}

\noindent
We shall be interested in the time-evolution of the states in the canonical domain (Schr\"{o}dinger picture), or equivalently, of the operator $\operator{T}$ (Heisenberg picture). We consider the unitary time-evolution operator $\operator{U}(t) = \exp(-i\operator{H}t/\hbar)$ and obtain the times $t$ of its corresponding time (parameter) invariant set $\tis$.

Given $\operator{T}(t) = \exp(i\operator{H}t/\hbar)\operator{T}\exp(-i\operator{H}t/\hbar)$, then the solution $\mathcal{C}(\operator{T}(t),\operator{H},\cd^+)$ or $\mathcal{C}(\operator{H},\operator{T}(t),\cd^-)$ for $t \in \tis$ has a time invariant set $\mathscr{T}$ of total measure zero (Theorem \ref{thm:tis}). From Proposition \ref{prop:tis}, it is of the form
\begin{equation}
    \tis = \qty{t = \frac{2n\pi\hbar}{\gcd(\tilde{\delta}(\operator{H}))}, n\in\mathbb{Z}}
\end{equation}
wherein $\tilde{\delta}(\operator{H})$ is the set of energy eigenvalue differences $E_k - E_l$ ($k>l$) except those differences containing $E_s$ wherein $\ket{s}\braket{s|\varphi} = 0$ for all $\ket{\varphi}\in\cd$. The time invariant set can then be modified by projecting $\operator{T}$ to a smaller subspace. We see that $\tis$ depends on the nature of the energy eigenvalues. If there exists two elements of $\tilde{\delta}(\operator{H})$ that are incommensurate, then $\tis$ is trivial; otherwise it is countably infinite. For all cases, $\tis$ has total measure zero.

It was previously believed that if the canonical relation is satisfied at $t = 0$, then it is satisfied for all times \cite{pauli1933handbuch,pauli1958handbuch,pauli1980general}. Here we see that this is not always true---indeed, the time operator and the Hamiltonian do not stay canonical as they evolve through time. Since $\operator{H}$ is bounded, we can write $\exp(-i\operator{H}t/\hbar) = \sum_{k=0}^\infty (k!)^{-1} (-i\operator{H}t/\hbar)^k $.  Let $t(\tau) = \tau + \tilde{t}$ where $\tau$ is a time interval $0\leq\tau<\tilde{t}$ and $\tilde{t}$ are the times in the time invariant set $\tis$. For short times, we then have $[\operator{T},\exp(-i\operator{H}t(\tau)/\hbar)] \ket{\varphi} = (\tau + \order{\tau^2})\ket{\varphi}$ for $\ket{\varphi} \in \cd^\pm$, giving us
\begin{equation} \label{eq:ldtime}
    \operator{T}(t(\tau)) \ket{\varphi} = \qty( \operator{T} \pm \tau + \order{\tau^2} ) \ket{\varphi} .
\end{equation}
Thus, for small $\tau$, i.e., for times near $\tis$, $\operator{T}(\tau) = \operator{T} \pm \tau$, and we get the expected equation $\dv*{\operator{T}\ket{\varphi}}{\tau} = \pm\ket{\varphi}$, and the canonical relation $[\operator{T}(\tau),\operator{H}]\ket{\varphi} = \pm i\hbar\ket{\varphi}$ is satisfied. This result is similar with the characteristic time operator in infinite-dimensional Hilbert space \cite{farrales2025characteristic}.

Thus, Hermitian time operators in finite-dimensional Hilbert space are never covariant, and only satisfy the expected dynamics in the neighborhood of a set of measure zero. They do not satisfy the usual definition of the canonical commutation relation nor the Weyl relation. They actually satisfy a weaker relation defined in \cite{arai2005generalized,arai2008time}.

\begin{definition}
    The pair of Hermitian operators $(\operator{T},\operator{H})$ satisfy the generalized weak Weyl relation in the Hilbert space $\hilbertspace$ if 
    \begin{equation}
        \operator{T} e^{-i\operator{H}t/\hbar} \ket{\psi} = e^{-i\operator{H}t/\hbar}(\operator{T} + \operator{K}(t))\ket{\psi} ,
    \end{equation}
    for all $\ket{\psi}$ in the domain of $\operator{T}$ and for all $t \in \reals$. The Hermitian operator $\operator{K}(t)$ is called the commuting factor.
\end{definition}

\noindent
Similar to the characteristic time operator \cite{arai2008time}, the commuting factor for the time operator in the form of \eqref{eq:ndcpa} is given by
\begin{equation}
    \operator{K}(t)\ket{\psi} = \sum_{s=1}^N \sum_{\substack{s' = 1 \\ s' \neq s}}^N \frac{i\hbar\psi_{s'}}{E_s - E_{s'}} \qty(e^{i(E_s - E_{s'})t/\hbar} - 1) \ket{s} ,
\end{equation}
where $\ket{\psi} = \sum_{s=1}^N \psi_s \ket{s} \in D(\operator{T}) = \hilbertspace$. All Hermitian time operators in finite-dimensional Hilbert space are thus generalized (weak Weyl) operators. We see that if $\ket{\psi}$ is in the canonical domain, then $\operator{K}'(0)\ket{\psi} = \ket{\psi}$, in accordance with \eqref{eq:ldtime}. This property is true for all time operators in finite-dimensional Hilbert space for states in the canonical domain.

In \cite{garrison1970canonically}, a quantum clock was defined to be a quantum system such that there exists a self-adjoint operator $\operator{T}$ and a state $\ket{\varphi(\tau)}$ such that $\exp(-i\operator{H}t/\hbar)\ket{\varphi(\tau)}=\ket{\varphi(\tau+t)}$ and $\braket{\varphi(\tau)|\operator{T}\varphi(\tau)}=\tau$ such that $\Delta_\varphi \operator{T}$ is negligible. States $\ket{\varphi(\tau)}$ satisfying this property are called \textit{clock states}. Since $\operator{T}$ and $\operator{H}$ form a canonical pair, then from Section \ref{sec:uncertainty}, they satisfy the \textit{time-energy uncertainty relation}
\begin{equation} \label{eq:uncreltime}
    (\Delta_\varphi \operator{T})(\Delta_\varphi \operator{H}) \geq \frac{\hbar}{2} \, ,
\end{equation}
for $\ket{\varphi}\in\cd$. There are many possible interpretations of such a relation \cite{busch1990a,busch1990b}, wherein $\Delta \operator{T}$ is usually taken to be some time duration. Here, time is treated as an observable, and thus, this uncertainty relation is interpreted the same way as the position-momentum uncertainty relation, i.e., $\Delta_\varphi \operator{T}$ is the standard deviation of measurements of $\operator{T}$ on an ensemble of states $\ket{\varphi}$.

For the quantum clocks considered here, $\Delta_\varphi \operator{T}$ need not be negligible since one only needs to increase the ensemble of measurements to obtain $\braket{\varphi(\tau)|\operator{T}\varphi(\tau)}$. Nevertheless, it is always possible to have $\Delta_\varphi \operator{T}$ be arbitrarily small by increasing $\Delta_\varphi \operator{H}$.

\begin{definition}
    Given the solution $\mathcal{C}(\operator{T},\operator{H},\cd^+)$ or $\mathcal{C}(\operator{H},\operator{T},\cd^-)$, a state $\ket{\varphi(t)}=\operator{U}(t)\ket{\varphi}$ is called a clock state if $\ket{\varphi} \in \cd^\pm$ and $t$ is in the neighborhood of $\tis$.
\end{definition}

\noindent
From Proposition \ref{prop:everyN}, we then see that any element of the Hilbert space can be a clock state corresponding to some appropriate choice of canonical pairs $(\operator{T},\operator{H})$.

\begin{remark}
    Measurement of $\operator{T}$ with respect to the clock state $\ket{\varphi(t)}$ provides the expectation value
    \begin{equation} \label{eq:evtime}
        \braket{\varphi(t(\tau))|\operator{T}\varphi(t(\tau))} = t_0 \pm \tau \, ,
    \end{equation}
    where $t(\tau) = \tau + \tilde{t}$ for small $\tau$, for all $\tilde{t} \in \tis$, and where $t_0 = \braket{\varphi(t=0)|\operator{T}\varphi(t=0)}$ is some constant.
\end{remark}

\noindent
The constant $t_0$ can then be thought of as some initial time, or a shift in the origin of time.

Thus, any quantum system with a finite-dimensional Hilbert space can act as a quantum clock, as once its Hamiltonian is identified, one can always find a \textit{characteristic} time operator, acting as the \textit{hand} of the clock (or a clock pointer). Similar with the characteristic time operator in infinite-dimensional Hilbert space \cite{farrales2025characteristic}, measurement of this time operator allows one to obtain parametric time in the neighborhood of the time invariant set $\tis$.

Inequivalent solutions are interpreted as two different canonical pairs that are physically different, and thus satisfy different properties. For example, given two inequivalent time operators $\operator{T}$ and $\operator{T}'=\operator{T}+\operator{F}$, the initial time for each operator $t_0 = \braket{\varphi|\operator{T}\varphi}$ and $t_0' = \braket{\varphi|\operator{T}'\varphi} = t_0 + \braket{\varphi|\operator{F}\varphi}$ are not equal if $\braket{\varphi|\operator{F}\varphi}\neq 0$. They may correspond to two different clocks that have different starting times.

Meanwhile, the uncertainty relation \eqref{eq:uncreltime} is followed by any time operator, i.e., it is satisfied by all equivalent and inequivalent solutions to the time-energy canonical commutation relation. According to Section \ref{sec:uncertainty}, the uncertainty relation is saturated only for some specific conditions, one of which is that $\operator{T}$ should only depend on two distinct eigenvalues of $\operator{H}$.

Different essentially canonical pairs also give different types of time operators. For example, given the Hamiltonian $\operator{H}$, an operator $\operator{T}^+$ satisfying $[\operator{T}^+,\operator{H}]\ket{\varphi}=i\hbar\ket{\varphi}$ for $\ket{\varphi}\in\cd^+$ gives $\dv*{\operator{T}^+}{t}=+1$ for $t \in \tis$, while an operator $\operator{T}^-$ satisfying $[\operator{T}^-,\operator{H}]\ket{\varphi}=-i\hbar\ket{\varphi}$ for $\ket{\varphi}\in\cd^-$ gives $\dv*{\operator{T}^-}{t}=-1$ for $t \in \tis$. The former corresponds to the canonical pair $(\operator{T},\operator{H})$, while the latter corresponds to the essentially canonical $(\operator{T},\operator{H})$, or the canonical $(\operator{H},\operator{T})$. We see that different essentially canonical pairs also correspond to different physical systems. Their properties differ, the former being an operator that increases in step with parametric time, the latter being an operator that decreases in step with time. $\operator{T}^+$ are then referred to as \textit{passage-time-type} time operators while $\operator{T}^-$ are referred to as \textit{time-of-arrival-type} time operators \cite{caballar2009characterizing}. For times near the time invariant set, the time operator expectation value becomes linear in time, i.e., $\braket{\varphi|\operator{T}^\pm(t(\tau))\varphi}=t_0\pm\tau$. Thus, in the context of being a quantum clock, $\operator{T}^+$ can be thought of as a \textit{quantum stopwatch}, while $\operator{T}^-$ is a \textit{quantum timer}.

Let us look at the two-dimensional time operator as an example \cite{SPP-2025-PB-14}. We can rewrite the canonical pair \eqref{eq:2dcp} as
\begin{equation}
    \operator{T} = \cos(\alpha) \frac{\hbar}{E_2 - E_1} \sigma_x - \sin(\alpha) \frac{\hbar}{E_2 - E_1} \sigma_y + \frac{1}{2}(a_1 - a_2) \sigma_z + \frac{1}{2} (a_1 + a_2) \operator{I}_\hilbertspace \, ,
\end{equation}
\begin{equation}
    \operator{H} = \frac{1}{2}(E_1 - E_2) \sigma_z + \frac{1}{2}(E_1 + E_2) \operator{I}_\hilbertspace \, ,
\end{equation}
where $\sigma_k$ are the Pauli matrices, and $\operator{I}_\hilbertspace$ is the identity. Note that the first two terms of $\operator{T}$ do not commute with $\operator{H}$, while the last two terms of $\operator{T}$ commute with $\operator{H}$. The canonical domain is given by \eqref{eq:2dcd}, i.e., $\cd^+ = \qty{ 2^{-1/2} \smqty(1 \\ ie^{-i\alpha}) }$. Adjusting $\alpha$ gives equivalent time operators, while adjusting $a_1$ and $a_2$ gives equivalent or inequivalent time operators. Time-of-arrival-type time operators are obtained by either considering $-\operator{T}$ in the same $\cd^+$, or by considering the same $\operator{T}$ but in $\cd^-=\qty{2^{-1/2}\smqty(1 \\ -ie^{-i\alpha})}$. We then see that the same pair $(\operator{T},\operator{H})$ defines two types of time operators depending on whether the state is in $\cd^+$ or $\cd^-$: if $\ket{\varphi_+}\in\cd^+$, then $\operator{T}$ is passage-time-type, while if $\ket{\varphi_-}\in\cd^-$, then $\operator{T}$ is time-of-arrival-type.

Let $\ket{\varphi_\pm} \in \cd^\pm$, i.e., $\ket{\varphi_\pm} = 2^{-1/2}\smqty(1 \\ \pm ie^{-i\alpha})$. For $\operator{U}(t) = \exp(-i\operator{H}t/\hbar)$, then $\operator{U}(t)\ket{\varphi}$ is in $\cd$ for all $t$ in the time invariant set $\tis = \qty{t = 2n\pi\hbar/E_{12}, n \in \mathbb{Z}}$. Let $\operator{T}(t) = \operator{U}^\dagger(t)\operator{T}\operator{U}(t)$. We then obtain $\braket{\varphi_\pm|\operator{T}(t(\tau))\varphi_\pm} = (a_1+a_2)/2 \pm (\hbar/E_{12}) \sin\qty(E_{12}\tau/\hbar)$. For small $\tau$, we have
\begin{equation}
    \braket{\varphi_\pm|\operator{T}(t(\tau))\varphi_\pm} = t_0 \pm \tau \, ,
\end{equation}
\begin{equation}
    \Delta_{\varphi_\pm} \operator{T}(t(\tau)) \Delta_{\varphi_\pm} \operator{H}(t(\tau)) = \frac{\hbar}{2} \sqrt{1 + \frac{(a_1-a_2)^2 E_{12}^2}{4\hbar^2}} \, ,
\end{equation}
where $t_0 = (a_1+a_2)/2$ and $t(\tau) = \tau + 2n\pi\hbar/E_{12}$ for $n \in \mathbb{Z}$, where the $2n\pi\hbar/E_{12}$ term are elements of the time invariant set $\tis$ for all $n$. The time invariant set is the same for all equivalent and inequivalent two-dimensional time operators.

With $\omega = E_{12}/\hbar$, the times in $\tis$ are integer multiples of the period $2\pi/\omega$, i.e., the time it takes for $\ket{\varphi_\pm}$ to return to itself. Thus, for small $\tau$, i.e., for $\omega\tau\ll1$, then $\sin(\omega\tau)\approx\omega\tau$ and the higher order terms $\order{\tau^2}$ vanish in the expectation value and uncertainty relation. In this scenario, the expectation value of $\operator{T}(t(\tau))$ linearly depends on $\tau$ and the uncertainty relation is always greater than or equal to $\hbar/2$. Note how equivalent solutions maintain the same properties, having the same expectation value and uncertainty relation. Inequivalent solutions affect the initial time $t_0$ in the expectation value, as well as increasing the uncertainty relation as $a_1$ goes away from $a_2$. In accordance with Section \ref{sec:uncertainty}, the uncertainty relation is indeed saturated when $a_1 = a_2$. Finally, the essentially canonical solutions arising from using either $\cd^\pm$ affects whether the expectation value increases or decreases with respect to time, while the uncertainty is unaffected. We then see how the clock behaves either as a quantum stopwatch or a quantum timer depending on the clock state used. All these different solutions constitute the infinitely-many possible time operators in two-dimensional Hilbert space.

As another example, consider the three-dimensional Hilbert space, with the Hamiltonian having two distinct eigenvalues, wherein one of them is doubly-degenerate. Consider then the time operator and Hamiltonian as a canonical pair in the form of \eqref{eq:3dcpdegen} with the one-dimensional canonical domain \eqref{eq:3dcddegen}. Let $t = \tau + \tilde{t}$ where $\tilde{t}$ is in its time invariant set $\tis = \qty{t = 2n\pi\hbar/E_{12},n\in\mathbb{Z}}$. For $\ket{\varphi_\pm} = 2^{-1} \smqty( \mp 1 \\ \mp 1 \\ \sqrt{2}) \in \cd^\pm$ and for small $\tau$, the expectation value and uncertainty relation is 
\begin{equation}
    \braket{\varphi_\pm|\operator{T}(t(\tau))\varphi_\pm} = t_0 \pm \tau \, ,
\end{equation}
\begin{equation}
    \Delta_{\varphi_\pm}\operator{T}(t(\tau)) \Delta_{\varphi_\pm}\operator{H}(t(\tau)) = \frac{\hbar}{2}\sqrt{1 + \frac{E_{12}^2}{4\hbar^2}(a_1^2+a_2^2+2a_3^2+2|b|^2+2(a_1+a_2)\Re(b) - 4t_0^2)} \, ,
\end{equation}
where $t_0 =  4^{-1}(a_1 + a_2 + 2 a_3 + 2 \Re(b))$. Similar to the previous example, we see that for short times, the expectation value of the time operator linearly depends with parametric time, and when $a_1=a_2=a_3$ and $b=0$, the uncertainty relation is saturated.

In general, given a time operator in an $N$-dimensional Hilbert space, we know that at times near the time invariant set, its expectation value \eqref{eq:evtime} becomes linear in time, and its uncertainty relation \eqref{eq:uncreltime} is always greater than or equal to $\hbar/2$. Projecting it to a smaller subspace preserves the dynamics, though it may affect the uncertainty. If projected to a subspace wherein there are only two distinct energy eigenvalues, the uncertainty can be minimized under certain conditions.

\section{Conclusions} \label{sec:conclusions}
Finite-dimensional quantum mechanics is filled with (essentially) canonical pairs and their corresponding uncertainty relations. Even the noncommuting Pauli matrices in two-dimensional Hilbert space are essentially canonical pairs by our definition, satisfying the relation $[\sigma_j,\sigma_k]\ket{\varphi_{l,\pm}}= 2i\epsilon_{jkl}\ket{\varphi_{l,\pm}}$, where $\epsilon_{jkl}$ is the Levi-Civita symbol (wherein $(\sigma_j,\sigma_k)$ is an essentially canonical pair only if $\epsilon_{jkl} \neq 0$), and where $\ket{\varphi_{l,\pm}}$ is in the one-dimensional canonical domain $\cd^\pm_l = \qty{\ket{\varphi_{l,\pm}} = \ket{\sigma_l,\pm}}$, where $\sigma_l\ket{\sigma_l,\pm} = \pm\ket{\sigma_l,\pm}$. This provides a natural way of generating \textit{canonical} pairs in two-dimensional Hilbert space through some appropriate mapping. Measurement of time can then be carried out by a measurement of a compatible observable, for example, an operator that is directly proportional to one of the Pauli matrices. In a Larmor clock for example, one can measure the spin to access time, but the clock is only linear with time when it is near its period. The fact that the time parameter invariant set for a time operator is always a set of measure zero shows the limitations of measuring time in quantum mechanics.

Since any pair of noncommuting operators are essentially canonical pairs, and since every Hermitian essentially canonical pair can be mapped to a canonical pair, then this lends to the following statement: that \textit{any operator that does not commute with the Hamiltonian can be measured to help obtain parametric time}. Consider some Hermitian operator $\operator{T'}$ such that $[\operator{T}',\operator{H}]\ket{\varphi}=ic\ket{\varphi}$ for $\ket{\varphi}$ in its canonical domain and for real $c$. Time can then be obtained with $\tau = (\hbar/c)\braket{\operator{T}'}_\varphi$ in the neighborhood of the time invariant set. 

Therefore, as was noted earlier, we see that \textit{every finite-dimensional quantum system is a quantum clock}, since given the system Hamiltonian, there always exists a \textit{clock hand} observable, the measurement of which yields time, though only in some specific set of times. The violation of the covariance property, $\operator{T}(t) = \operator{T}(0) + t$ for all $t \in \mathbb{R}$, may be a consequence of the Hermiticity of the time operator. If one were to strictly abide with the belief that time operators have to be covariant, we will be forced to neglect the existence of a time observable in standard quantum mechanics for finite-dimensional quantum systems. The fact that it is still possible to access time using the Hermitian operators constructed here shows that covariance is not a necessary property of all time observables. These Hermitian operators are but one possible class of solutions to the canonical commutation relation. Another class of solutions wherein the time operator does obey the covariance property could be obtained, but sacrificing self-adjointness may be required.

The ideas here can be extended to the case when the Hilbert space is infinite-dimensional. For example, there has been much study on the quantization of the classical time of arrival for a particle in $\hilbertspace=L^2(I\subseteq\mathbb{R})$: the dynamics of its eigenfunctions was studied in \cite{galapon2018quantizations,pablico2024role}, and an application to the tunneling time problem was studied in \cite{galapon2012only,flores2024partial,flores2024instantaneous}. In quantization, a classical phase space function is mapped to its quantum image by promoting position and momentum to their respective quantum operator counterparts; but there exists infinitely many quantization rules, and a particular rule does not necessarily lead to an operator which satisfies the usual definition of the canonical commutation relation \cite{galapon2018quantizations,pablico2024role,gotay2000obstructions}. This has been determined as the chief weakness of quantization as a method of constructing time operators. Supraquantization can be used as an alternative method, with the canonical relation already being a requirement in the construction of time operators \cite{farrales2022conjugates,pablico2023quantum,pablico2024moyal,galapon2004shouldnt}. However, by being cognizant of the canonical domain, it may be that the quantized time operators (or in fact, any operator not commuting with the Hamiltonian) satisfy the time-energy canonical commutation relation in a proper subspace of the Hilbert space. One only needs to solve for the eigenvalues of the corresponding commutator and obtain the eigenspaces, thus determining the canonical domain. In this subspace, the operator is linear in time, and measurement yields this time parameter when near the time invariant set. In this view, it may be possible to accommodate the quantized time-of-arrival operators as forming a canonical pair with the Hamiltonian, wherein the canonical domain may depend on the quantization rule used, with different quantization rules providing different equivalent or inequivalent solutions.

\begin{appendices}

\section{Solutions to the three-dimensional canonical commutation relation} \label{app:3dsoln}
\subsection{Nondegenerate case} \label{app:3dnondegen}

The general matrix
\begin{equation}
    \operator{A} = \mqty(a_1 && \beta_{12}\frac{\hbar}{B_{12}}e^{i\alpha_{12}} && \beta_{13}\frac{\hbar}{B_{13}}e^{i\alpha_{13}} \\ \beta_{12}^*\frac{\hbar}{B_{12}}e^{-i\alpha_{12}}  && a_2 && \beta_{23}\frac{\hbar}{B_{23}}e^{i\alpha_{23}} \\ \beta_{13}^*\frac{\hbar}{B_{13}}e^{-i\alpha_{13}} && \beta_{23}^*\frac{\hbar}{B_{23}}e^{-i\alpha_{23}}  && a_3) \, ,
\end{equation}
(real $a_1$, $a_2$, $a_3$, $\alpha_{12}$, $\alpha_{13}$, and $\alpha_{23}$; complex $\beta_{12}$, $\beta_{13}$, and $\beta_{23}$) is an essentially canonical pair with the non-degenerate $\operator{B}$ given the following cases:
\begin{enumerate}
    \item $\beta_{12} \neq 0$, $\beta_{13} \neq 0$, $\beta_{23} \neq 0$, let $\abs{\beta}^2 = \abs{\beta_{12}}^2 + \abs{\beta_{13}}^2 + \abs{\beta_{23}}^2$; let $\phi_1$, $\phi_2$ and $\phi$ be complex, whose values are such that the elements of the canonical domain are normalized.
    \begin{enumerate}
        \item $\abs{\beta}^2 = 3$, $\Im\qty(\beta_{12}\beta_{13}^*\beta_{23}e^{i(\alpha_{12}-\alpha_{13}+\alpha_{23})}) = -1$ \begin{itemize}
            \item $[\operator{A},\operator{B}]\ket{\varphi} = i\hbar\ket{\varphi}$ with $\cd = \qty{\ket{\varphi} = \phi_1 \smqty(1 \\ 0 \\ i\beta_{13}^* \exp(-i\alpha_{13})) + \phi_2 \smqty(0 \\ 1 \\ i\beta_{23}^*\exp(-i\alpha_{23}))}$
            \item $[\operator{A},\operator{B}]\ket{\varphi} = -2i\hbar\ket{\varphi}$ with $\cd = \qty{\ket{\varphi} = \phi \smqty(1 \\ \frac{\abs{\beta_{13}}^2-4}{-2i\beta_{12}\exp(i\alpha_{12})-\beta_{13}\beta_{23}^*\exp(-i(-\alpha_{13}+\alpha_{23}))} \\ \frac{\abs{\beta_{12}}^2 - 4}{-2i\beta_{13}\exp(i\alpha_{13})+\beta_{12}\beta_{23}\exp(i(\alpha_{12}+\alpha_{23}))})}$
        \end{itemize}
        \item $\abs{\beta}^2 \neq 3$, $\abs{\beta}^2>\frac{3}{4}$, $\Im\qty(\beta_{12}\beta_{13}^*\beta_{23}e^{i(\alpha_{12}-\alpha_{13}+\alpha_{23})}) = -(\abs{\beta}^2-1)/2$
        \begin{itemize}
            \item $[\operator{A},\operator{B}]\ket{\varphi} =i\hbar\ket{\varphi}$ with $\cd = \qty{\ket{\varphi} = \phi \smqty(1 \\ \frac{\abs{\beta_{13}}^2-1}{i\beta_{12}\exp(i\alpha_{12})-\beta_{13}\beta_{23}^*\exp(-i(-\alpha_{13}+\alpha_{23}))} \\ \frac{\abs{\beta_{12}}^2 - 1}{i\beta_{13}\exp(i\alpha_{13})+\beta_{12}\beta_{23}\exp(i(\alpha_{12}+\alpha_{23}))})}$
            \item $[\operator{A},\operator{B}]\ket{\varphi} =-\frac{1}{2}(1+\sqrt{4\abs{\beta}^2-3})i\hbar\ket{\varphi}$ with 
            \begin{equation}
                \cd = \qty{\ket{\varphi} = \phi \smqty(1 \\ \frac{\abs{\beta_{13}}^2- 4^{-1}(1+\sqrt{4\beta-3})^2 }{-2^{-1}(1+\sqrt{4\beta-3})i\beta_{12}\exp(i\alpha_{12})-\beta_{13}\beta_{23}^*\exp(-i(-\alpha_{13}+\alpha_{23}))} \\ \frac{\abs{\beta_{12}}^2 - 4^{-1}(1+\sqrt{4\beta-3})^2}{-2^{-1}(1+\sqrt{4\beta-3})i\beta_{13}\exp(i\alpha_{13})+\beta_{12}\beta_{23}\exp(i(\alpha_{12}+\alpha_{23}))})} \nonumber
            \end{equation}
            \item $[\operator{A},\operator{B}]\ket{\varphi} =-\frac{1}{2}(1-\sqrt{4\abs{\beta}^2-3})i\hbar\ket{\varphi}$ with
            \begin{equation}
                \cd = \qty{\ket{\varphi} = \phi \smqty(1 \\ \frac{\abs{\beta_{13}}^2- 4^{-1}(1-\sqrt{4\beta-3})^2 }{-2^{-1}(1-\sqrt{4\beta-3})i\beta_{12}\exp(i\alpha_{12})-\beta_{13}\beta_{23}^*\exp(-i(-\alpha_{13}+\alpha_{23}))} \\ \frac{\abs{\beta_{12}}^2 - 4^{-1}(1-\sqrt{4\beta-3})^2}{-2^{-1}(1-\sqrt{4\beta-3})i\beta_{13}\exp(i\alpha_{13})+\beta_{12}\beta_{23}\exp(i(\alpha_{12}+\alpha_{23}))})} \nonumber
            \end{equation}
        \end{itemize}
    \end{enumerate}
    \item \begin{enumerate}
        \item $\beta_{12} = 0$, $\abs{\beta_{13}}^2 + \abs{\beta_{23}}^2 = 1$, and $\alpha_{13},\alpha_{23}\in\mathbb{R}$
        \begin{itemize}
            \item $[\operator{A},\operator{B}]\ket{\varphi} = i\hbar\ket{\varphi}$ with $\cd = \qty{\ket{\varphi} = 2^{-1/2} \smqty(i\beta_{13} \exp(i\alpha_{13}) \\ i\beta_{23}\exp(i\alpha_{23}) \\ 1)}$
            \item $[\operator{A},\operator{B}]\ket{\varphi} = -i\hbar\ket{\varphi}$ with $\cd = \qty{\ket{\varphi} = 2^{-1/2} \smqty(-i\beta_{13} \exp(i\alpha_{13}) \\ -i\beta_{23}\exp(i\alpha_{23}) \\ 1)}$
            \item $[\operator{A},\operator{B}]\ket{\varphi} = 0$ with $\cd = \qty{\ket{\varphi} = \smqty( \beta_{23}^*\exp(-i\alpha_{23}) \\ -\beta_{13}^*\exp(-i\alpha_{13}) \\ 0 )}$
        \end{itemize}
        \item $\beta_{13} = 0$, $\abs{\beta_{12}}^2 + \abs{\beta_{23}}^2 = 1$, and $\alpha_{12},\alpha_{23}\in\mathbb{R}$ \begin{itemize}
            \item $[\operator{A},\operator{B}]\ket{\varphi} = i\hbar\ket{\varphi}$ with $\cd = \qty{\ket{\varphi} = 2^{-1/2} \smqty(i\beta_{12} \exp(i\alpha_{12}) \\ 1 \\ -i\beta_{23}^* \exp(-i\alpha_{23}))}$
            \item $[\operator{A},\operator{B}]\ket{\varphi} = -i\hbar\ket{\varphi}$ with $\cd = \qty{\ket{\varphi} = 2^{-1/2} \smqty(-i\beta_{12} \exp(-i\alpha_{12}) \\ 1 \\ i\beta_{23}^* \exp(i\alpha_{23}))}$
            \item $[\operator{A},\operator{B}]\ket{\varphi} = 0$ with $\cd = \qty{\ket{\varphi} = \smqty( \beta_{23}\exp(i\alpha_{23}) \\ 0 \\ -\beta_{12}^* \exp(-i\alpha_{12}) )}$
        \end{itemize}
        \item $\beta_{23} = 0$, $\abs{\beta_{12}}^2 + \abs{\beta_{13}}^2 = 1$, and $\alpha_{12},\alpha_{13}\in\mathbb{R}$ \begin{itemize}
            \item $[\operator{A},\operator{B}]\ket{\varphi} = i\hbar\ket{\varphi}$ with $\cd = \qty{\ket{\varphi} = 2^{-1/2} \smqty( 1 \\ -i\beta_{12}^* \exp(-i\alpha_{12}) \\ -i\beta_{13}^* \exp(-i\alpha_{13}) )}$
            \item $[\operator{A},\operator{B}]\ket{\varphi} = -i\hbar\ket{\varphi}$ with $\cd = \qty{\ket{\varphi} = 2^{-1/2} \smqty( 1 \\ i\beta_{12}^* \exp(-i\alpha_{12}) \\ i\beta_{13}^* \exp(-i\alpha_{13}) )}$
            \item $[\operator{A},\operator{B}]\ket{\varphi} = 0$ with $\cd = \qty{\ket{\varphi} = \smqty( 0 \\ \beta_{13} \exp(i\alpha_{13}) \\ -\beta_{12} \exp(i\alpha_{12}) )}$
        \end{itemize}
    \end{enumerate}
\end{enumerate}
Note that there are cases when the commutator is equal to zero even in some nontrivial subspace of the Hilbert space. By definition, they are not essentially canonical pairs, but they are still shown here for completeness. 

The solutions can be thought of as $\beta_{12}\operator{A}_{12} + \beta_{13}\operator{A}_{13} + \beta_{23} \operator{A}_{23}$ where $\operator{A}_{kl} = \operator{P}_{kl}\operator{A}\operator{P}_{kl}$, $\operator{P}_{kl} = \op{k} + \op{l}$, with $\operator{A}$ as given in \eqref{eq:3dcp}. Thus, from $\operator{A}$, one can obtain all possible solution by varying $\alpha_{kl}$ (giving equivalent solutions), $a_k$ (which may give equivalent or inequivalent solutions), or by taking a sum of its projections (giving inequivalent solutions).

\subsection{Degenerate case} \label{app:3ddegen}
The general matrix
\begin{equation}
    \operator{A} = \mqty(a_1 && \beta_{12} && \beta_{13}\frac{\hbar}{B_{12}}e^{i\alpha_{13}} \\ \beta_{12}^* && a_2 && \beta_{23}\frac{\hbar}{B_{12}}e^{i\alpha_{23}} \\ \beta_{13}^*\frac{\hbar}{B_{12}}e^{-i\alpha_{13}} && \beta_{23}^*\frac{\hbar}{B_{12}}e^{-i\alpha_{23}}  && a_3) \, ,
\end{equation}
(real $a_1$, $a_2$, $a_3$, $\alpha_{13}$, and $\alpha_{23}$; complex $\beta_{12}$, $\beta_{13}$, and $\beta_{23}$)  is an essentially canonical pair with $\operator{B} = \smqty(B_1 && 0 && 0 \\ 0 && B_1 && 0 \\ 0 && 0 && B_2)$ with doubly-degenerate $B_1$ and nondegenerate $B_2$ when $\abs{\beta_{13}}^2 + \abs{\beta_{23}}^2 = 1$ with the following cases:
\begin{itemize}
    \item $[\operator{A},\operator{B}]\ket{\varphi} = i\hbar\ket{\varphi}$ with $\cd = \qty{\ket{\varphi} = 2^{-1/2} \smqty( -i\beta_{13} \exp(i\alpha_{13}) \\ -i\beta_{23} \exp(i\alpha_{23}) \\ 1 )}$
    \item $[\operator{A},\operator{B}]\ket{\varphi} = -i\hbar\ket{\varphi}$ with $\cd = \qty{\ket{\varphi} = 2^{-1/2} \smqty( i\beta_{13} \exp(i\alpha_{13}) \\ i\beta_{23} \exp(i\alpha_{23}) \\ 1 )}$
    \item $[\operator{A},\operator{B}]\ket{\varphi} = 0$ with $\cd = \qty{\ket{\varphi} = \smqty( \beta_{23}^* \exp(i\alpha_{23}) \\ - \beta_{13}^* \exp(-i\alpha_{13}) \\ 0 )}$
\end{itemize}

\end{appendices}


\end{document}